\documentclass[11pt]{amsart}
\usepackage{hyperref}

\usepackage{graphicx}
\usepackage{amsmath,amssymb,mathrsfs}
\usepackage{amsthm}

\newtheorem{thm}{Theorem}[section]
\newtheorem*{thm*}{Theorem}
\newtheorem{cor}[thm]{Corollary}
\newtheorem{lem}[thm]{Lemma}
\newtheorem{prop}[thm]{Proposition}

\theoremstyle{definition}
\newtheorem{defn}[thm]{Definition}
\theoremstyle{remark}
\newtheorem{rem}[thm]{Remark}
\theoremstyle{example}
\newtheorem{exa}[thm]{Example}
\theoremstyle{conjecture}

\numberwithin{equation}{section}

\newcommand{\abs}[1]{\left\vert#1\right\vert}

\newcommand{\calA}{\mathcal A}
\newcommand{\calC}{\mathcal C}
\newcommand{\calF}{\mathcal F}
\newcommand{\calH}{\mathcal H}

\newcommand{\calJ}{\mathcal J}
\newcommand{\calK}{\mathcal K}

\newcommand{\calD}{\mathcal D}
\newcommand{\calQ}{\mathcal Q}

\newcommand {\C} {\mathbb C}
\newcommand {\R} {\mathbb R}
\newcommand {\X} {\mathbb X}

\newcommand {\N} {\mathbb N}

\begin{document}

\title[Differential Operators]{Complex Symmetric Differential Operators on Fock Space}%

\date{\today}%

\author{Pham Viet Hai}%
\address[P. V. Hai]{National University of Singapore, Singapore}%
\email{isepvh@nus.edu.sg}

\author{Mihai Putinar}
\address[M.~Putinar]{University of California at Santa Barbara, CA,
USA and Newcastle University, Newcastle upon Tyne, UK} 
\email{\tt mputinar@math.ucsb.edu, mihai.putinar@ncl.ac.uk}




\maketitle

\begin{abstract}
The space of entire functions which are integrable with respect to the Gaussian weight, known also as the Fock space, is one of the preferred
functional Hilbert spaces for modeling and experimenting harmonic analysis, quantum mechanics or spectral analysis phenomena. This space of entire functions carries
a three parameter family of canonical isometric involutions. We characterize the linear differential operators acting on Fock space which are complex symmetric
with respect to these conjugations. In parallel, as a basis of comparison, we discuss the structure of self-adjoint linear differential operators. The computation of the point spectrum of some of these operators is carried out in detail.
\end{abstract}

\section{Introduction}

\subsection{Complex symmetric operators}
The foundational works of Glazman and Zihar \cite{IMG1, IMG2, Zihar} mark the beginning of the theory of unbounded complex symmetric operators. The non-symmetric or non-self-adjoint differential operators treated by them
did not enter into the classical framework initiated by von Neumann, nor in the dissipative operator class studied around mid XX-th century by Keldys, Krein and Livsic. What makes complex symmetric differential operators special is the fact that they carry a weak form of spectral decomposition theorem (discovered much earlier in the case of matrices by Takagi). While the main path of research was continued after Glazman on a Hamiltonian mechanics path, in the context of mature by now operator theory in spaces with an indefinite metric \cite{Azizov, Langer-Tretter}, the more recent theory of non-hermitian quantum mechanics
has naturally enlarged the class of examples of complex symmetric differential operators \cite{MR1686605, Caliceti, Henry, Znojil}. Pertinent studies of complex symmetric boundary conditions for linear differential operators on an interval are also well known \cite{Knowles, Race}. Adding to these examples classes of structured matrices (such as Toeplitz or Hankel) or integral operators which carry a canonical complex symmetry,
we are contemplating today a vast territory, escaping the orthodox path of spectral analysis, but by no means less captivating and relevant for applications. A systematic study, at the abstract level, of complex symmetric operators was undertaken in \cite{GP1, GP2}.  Applications to mathematical physics and connections to the theory of operators in an indefinite metric space can be found in the survey \cite{GPP}. 

The present article narrows down to the precise question of classifying
linear differential operators acting on Fock space, which are complex symmetric with respect to a family of natural conjugations. These conjugations include as a special case the famous $\mathcal{PT}$-symmetry.
Besides tedious but elementary identities at the level of formal series, the main complication of our study resides in the identification of the proper domains of the unbounded linear operators in question.

We start by recalling some basic definitions. Throughout this article $\calH$ is a separable complex Hilbert space. The domain of an unbounded linear operateor $T$ is $\text{dom}(T)$. For two unbounded operators $A$, $B$, the notation $A\preceq B$ means that $B$ is an \emph{extension} of $A$ or that $A$ is a \emph{restriction} of $B$ on $\text{dom}(A)$, namely $\text{dom}(A)\subseteq\text{dom}(B)$ and $Ax=Bx$ for every $x\in\text{dom}(A)$. If $A\preceq B$ and $B\preceq A$, then we write $A=B$. Furthermore, if $C,D$ are two bounded operators on $\calH$, then we define the operator $CAD$ by $(CAD)f=C(A(Df))$ with the domain $\text{dom}(CAD)$ including all elements $f\in\calH$ for which $Df\in\text{dom}(A)$. 

\begin{defn}
An anti-linear operator $\calC:\calH\to\calH$ is called a \emph{conjugation} if it is both involutive and isometric.
\end{defn}

In the presence of a conjugation $\calC$, the inner product of the underlying Hilbert space induces a bounded, complex bilinear symmetric form
$$ [x,y] = \langle x, {\calC} y \rangle, \ \ x,y \in {\calH}.$$
Symmetry of linear transforms with respect to this bilinear form is the main subject of our study.

\begin{defn}\label{defcso-new}
Let $T\colon\text{dom}(T)\subseteq\calH \to \calH$ be a closed, densely defined, linear operator and $\calC$ a conjugation. We say that $T$ is \begin{enumerate}
\item \emph{$\calC$-symmetric} if $T\preceq\calC T^*\calC$;
\item \emph{$\calC$-selfadjoint} if $T=\calC T^*\calC$.
\end{enumerate}

In both cases the operator $T$ is \emph{complex symmetric}, that is
$$ [Tx,y] = [x,Ty], \ \ x,y \in \text{dom}(T).$$
\end{defn}

Note that every $\calC$-selfadjoint operator is $\calC$-symmetric, but the converse implication is not always true as we shall see below.

Recently, complex symmetry was studied in itself on Hilbert spaces of \emph{holomorphic functions}. The first works \cite{GH,JKKL} were devoted to bounded \emph{weighted composition operators} 
$$
W_{\psi,\varphi}f=\psi\cdot f\circ\varphi
$$
on Hardy spaces with the standard conjugation 
$$\calJ f(z)=\overline{f(\overline{z})}.$$
The structure of $\calJ$ inspired the first author to investigate a generalization, namely \emph{anti-linear weighted composition operators}
$$
\calA_{\psi,\varphi}f=\psi\cdot\overline{f\circ\overline{\varphi}}
$$
acting on Fock space, in \cite{HK1}. These are the conjugations entering into the study carried into the present article.

\subsection{Fock space}
The \emph{Fock space} $\calF^2(\C)$ (sometimes called the Segal-Bargmann space) consists of entire functions which are square integrable with respect to the Gaussian measure $\frac{1}{\pi}e^{-|z|^2}\;dV(z)$, where $dV$ is the Lebesgue measure on $\C$. It is well-known that this is a reproducing kernel Hilbert space, with inner product
$$
\langle f,g\rangle=\dfrac{1}{\pi}\int_\C f(z)\overline{g(z)}e^{-|z|^2}\;dV(z),
$$
and kernel function
$$
K_z^{[m]}(u)=u^m e^{u\overline{z}},\quad m\in\N,z,u\in\C.
$$
To simplify notation we write $K_z$ in case $m=0$. These kernel functions satisfy 
$$
f^{(m)}(z)=\langle f,K_z^{[m]}\rangle,\quad\forall f\in\calF^2(\C).
$$
Since 
\begin{equation*}
\langle z^m,z^n\rangle=
\begin{cases}
0,\quad\text{if $m\ne n$},\\
n!,\quad\text{if $m=n$},
\end{cases}
\end{equation*}
the inner product in $\calF^2(\C)$ turns out to be diagonalizable:
\begin{equation}\label{inner-product}
\langle f,g\rangle=\sum_{n\geq 0}f_n\overline{g_n}n!,
\end{equation}
where $f_n$ and $g_n$ are the Taylor coefficients of $f$, $g$, respectively. 
Notable at this incipient level of discussion is the adjunction formula
$$ \langle z p, q \rangle = \langle p , \frac{\partial}{\partial z} q \rangle, \ \ p,q \in {\mathbb C}[z].$$
More information about Fock space can be found in the monograph \cite{KZ}.

A characterization of all anti-linear operators $\calA_{\psi,\varphi}$, which are conjugations on $\calF^2(\C)$ was carried out in \cite{HK1}. These anti-linear operations are called \emph{weighted composition conjugations}, and they can be described as follows. For complex numbers $a$, $b$, $c$ satisfying
\begin{equation}\label{abc-cond}
|a|=1, \quad \bar{a}b+\bar{b}=0, \quad |c|^2e^{|b|^2}=1,
\end{equation}
the weighted composition conjugation is defined by
\begin{equation}\label{Ca,b,c-wcc}
\calC_{a,b,c}f(z):=ce^{bz}\overline{f\left(\overline{az+b}\right)},\quad f\in\calF^2(\C).
\end{equation}

An important particular case is the $\mathscr{PT}$-symmetry:
$$ f(z) \mapsto \overline{f(-\overline{z})}.$$

\subsection{Differential operators}

Consider formal linear differential expressions of the form
$$
(\calD[\psi] f)(z)=\sum_{j=0}^\kappa\psi_j(z)f^{(j)}(z),
$$
where $\kappa$ is a non-negative integer, and $(\psi_j)_{j=0}^\kappa$ is a family of entire functions.

Denote by $\text{dom}(T_{\max})$ the set of all functions $f$ in $\calF^2(\C)$ for which $\calD[\psi] f\in\calF^2(\C)$. We define the operator $T_{\max}$ on $\calF^2(\C)$ by $T_{\max}f=\calD[\psi] f$. This operator is called the \emph{maximal differential operator} of order $\kappa$, corresponding to the formal expression $\calD[\psi]$. It is ''maximal" in the sense that it
cannot be extended as an operator in $\calF^2(\C)$ generated by $\calD[\psi]$.

The operator $T$ is called a \emph{$\kappa$th-order differential operator} if $T\preceq T_{\max}$, namely
$$
\text{dom}(T)\subseteq\text{dom}(T_{\max}),\quad Tf=\calD[\psi] f.
$$
In this situation, the domain $\text{dom}(T)$ is an arbitrary subspace of $\text{dom}(T_{\max})$, and the operator $T$ is the restriction of the maximal operator $T_{\max}$ on $\text{dom}(T)$.

\begin{exa}\label{Gamma12}
Let $a,b,c$ be complex constants satisfying condition \eqref{abc-cond}, and $G,K,\alpha$ arbitrary complex constants. We consider first-order differential operators $\Gamma_1,\Gamma_2$ given by
$$
\Gamma_1f(z)=(G+aKz)f(z)+Kf'(z),\quad\text{dom}(\Gamma_1)=\{f\in\calF^2(\C):\Gamma_1 f\in\calF^2(\C)\},
$$
$$
\Gamma_2g(z)=(G-\alpha(aK+b)z)g(z)+\alpha(z-K)g'(z),$$ $$\quad\text{dom}(\Gamma_2)=\{g\in\calF^2(\C):\Gamma_2 g\in\calF^2(\C)\}.
$$
One proves that the operators $\Gamma_1,\Gamma_2$ are $\calC$-selfadjoint with respect to the conjugation $\calC_{a,b,c}$ and moreover, they are generators of some complex symmetric semigroups on Fock space $\calF^2(\C)$, see \cite[Propositions 4.7-4.8]{HK4}.
\end{exa}

\subsection{Contents}
In Section \ref{sec-2}  the basic properties of a linear differential operator on Fock space are discussed: dense domain, closed graph, computation of the adjoint. In Proposition \ref{T-closed} we prove that every maximal differential operator is closed on $\calF^2(\C)$, while Theorem \ref{adjoint-form} provides the structure of the adjoint when the total symbol is polynomial. Proposition \ref{propT*K_z} unveils the action of the adjoint operator on the reproducing kernel elements. We characterize in Section \ref{c-self-self} \emph{maximal} differential operators which are $\calC$-selfadjoint with respect to a weighted composition conjugation (Theorem \ref{Cself}). A similar computation is carried out for selfadjoint (Theorem \ref{selfadjoint}) operators in the classical sense. Section \ref{non-maximal-diff-op} is devoted to linear differential operators with \emph{arbitrary} domain (not necessarily maximal). Theorems \ref{no-max-do}-\ref{no-max-do-self} show that there is no nontrivial domain for a differential operator $T$ on which $T$ is $\calC$-selfadjoint with respect to some weighted composition conjugation, and selfadjoint, respectively. Section 5 recalls the specific unitary equivalence between Fock space and Lebesgue $L^2$-space which intertwines the creation and annihilation operators; our aim is to transfer to and from the classical Lebesgue space some spectral analysis aspects of linear differential operators, including in particular the class of $\mathcal{PT}$-symmetric  operators. Section 6 is devoted to the computation of the point spectrum of certain complex symmetric differential operators, in the spirit, and generalizing the well known case of the quantum oscillator.

\section{Preliminaries}\label{sec-2}
This section contains several technical observations which will be later on referred to. Some of these statements may have an intrinsic value.
\subsection{Reproducing kernel algebra}
We are first concerned with the action of a differential operator on the kernel functions, that is on the point evaluation functionals.
\begin{prop}\label{propT*K_z}
Let $T$ be a $\kappa$th-order differential operator induced by $(\psi_j)_{j=0}^\kappa$ (note that $T\preceq T_{\max}$). Suppose that $T$ is densely defined. Then $K_z^{[m]}\in\text{dom}(T^*)$  for every $z\in\C$, $m\in\N$, and moreover,
\begin{equation}\label{T*K_z}
T^*K_z^{[m]}=\sum_{j=0}^{\kappa+m}\overline{\omega_{j,m}(z)}K_z^{[j]}.
\end{equation}
Here, the two-parameter family $(\omega_{j,\ell})$ of entire functions is recursively defined by
\begin{equation}\label{omega}
\omega_{j,\ell}(z)=
\begin{cases}
\psi_j(z),\quad\ell=0,0\leq j\leq\kappa,\\
\omega_{0,\ell-1}'(z),\quad \ell\geq 1,j=0,\\
\omega_{j,\ell-1}'(z)+\omega_{j-1,\ell-1}(z),\quad \ell\geq 1,1\leq j\leq\kappa+\ell-1,\\
\omega_{\kappa+\ell-1,\ell-1}(z),\quad \ell\geq 1,j=\kappa+\ell.
\end{cases}
\end{equation}
\end{prop}

\begin{proof} Induction on $m$. For $f\in\text{dom}(T)$ we find
\begin{eqnarray*}
\langle Tf,K_z\rangle &=& (Tf)(z)=\sum_{j=0}^\kappa\psi_j(z)f^{(j)}(z)=\langle f,\sum_{j=0}^\kappa\overline{\psi_j(z)}K_z^{[j]}\rangle,
\end{eqnarray*}
which gives \eqref{T*K_z} in the case when $m=0$.

Suppose that the conclusion is true for $m=\ell$, i.e. there exist entire functions $\omega_{j,\ell}$ such that
$$
T^*K_z^{[\ell]}=\sum_{j=0}^{\kappa+\ell}\overline{\omega_{j,\ell}(z)}K_z^{[j]}.
$$
Taking inner product both sides with the function $f$, we get
\begin{eqnarray*}
\langle f,T^*K_z^{[\ell]}\rangle =\sum_{j=0}^{\kappa+\ell}\omega_{j,\ell}(z)\langle f,K_z^{[j]}\rangle=\sum_{j=0}^{\kappa+\ell}\omega_{j,\ell}(z)f^{(j)}(z).
\end{eqnarray*}
On the other hand, by the definition of adjoint operators,
$$
\langle f,T^*K_z^{[\ell]}\rangle=\langle Tf,K_z^{[\ell]}\rangle=(Tf)^{(\ell)}(z).
$$
Therefore
$$
(Tf)^{(\ell)}(z)=\sum_{j=0}^{\kappa+\ell}\omega_{j,\ell}(z)f^{(j)}(z),
$$
which implies, by differentiating both sides, that
\begin{eqnarray*}
(Tf)^{(\ell+1)}(z)
&=&\omega_{0,\ell}'(z)f(z)+\sum_{j=1}^{\kappa+\ell}[\omega_{j,\ell}'(z)+\omega_{j-1,\ell}(z)]f^{(j)}(z)\\
&&+\omega_{\kappa+\ell,\ell}(z)f^{(\kappa+\ell+1)}(z)\\
&=& \sum_{j=0}^{\kappa+\ell+1}\omega_{j,\ell+1}(z)f^{(j)}(z)=\langle f,\sum_{j=0}^{\kappa+\ell+1}\overline{\omega_{j,\ell+1}(z)}K_z^{[j]}\rangle.
\end{eqnarray*}
Again by the very definition of kernel functions one finds
\begin{eqnarray*}
(Tf)^{(\ell+1)}(z) &=&\langle Tf,K_z^{[\ell+1]}\rangle.
\end{eqnarray*}
Therefore,
$$
\langle Tf,K_z^{[\ell+1]}\rangle=\langle f,\sum_{j=0}^{\kappa+\ell+1}\overline{\omega_{j,\ell+1}(z)}K_z^{[j]}\rangle.
$$
The above identity gives $K_z^{[\ell+1]}\in\text{dom}(T^*)$ and \eqref{T*K_z}.
\end{proof}

Next we study actions of $\calC_{a,b,c}T^*\calC_{a,b,c}$ on the kernel functions as a first step towards solving the symmetry condition $\calC_{a,b,c}T^*\calC_{a,b,c}=T$.
\begin{prop}\label{CT*CKz}
Let $T$ be a $\kappa$th-order differential operator induced by $(\psi_j)_{j=0}^\kappa$ (note that $T\preceq T_{\max}$), and $\calC_{a,b,c}$ a weighted composition conjugation. Suppose that $T$ is densely defined. Then for every $z\in\C$, we have 
\begin{enumerate}
\item $\calC_{a,b,c}K_z=ce^{bz}K_{\overline{az+b}}$.
\item $K_z\in\text{dom}(\calC_{a,b,c}T^*\calC_{a,b,c})$, and
$$
\calC_{a,b,c}T^*\calC_{a,b,c} K_z(u)=\sum_{j=0}^\kappa\psi_j(\overline{az+b})(au+b)^j e^{u\overline{z}}.
$$
\end{enumerate}
\end{prop}
\begin{proof}
Let $z,u\in\C$.

(1)Remark that
$$
\calC_{a,b,c} K_z(u) = ce^{bu}\overline{K_z(\overline{au+b})}=ce^{bu+(au+b)z}=ce^{bz}e^{(az+b)u}=ce^{bz}K_{\overline{az+b}}(u).
$$

(2) By the first part and anti-linearity of $\calC_{a,b,c}$, we have
$$
\calC_{a,b,c}T^*\calC_{a,b,c} K_z(u) = \calC_{a,b,c}T^*\left(ce^{bz}K_{\overline{az+b}}\right)(u)=\overline{c}e^{\overline{bz}}\calC_{a,b,c}T^*\left(K_{\overline{az+b}}\right)(u).
$$
By \eqref{T*K_z}, the last expression is equal to
$$
\overline{c}e^{\overline{bz}}\sum_{j=0}^\kappa\psi_j(\overline{az+b})\calC_{a,b,c}\left(K_{\overline{az+b}}^{[j]}\right)(u).
$$
Taking into account the form of $\calC_{a,b,c}$ in \eqref{Ca,b,c-wcc}, we infer
\begin{eqnarray*}
\calC_{a,b,c}T^*\calC_{a,b,c} K_z(u)
&=&|c|^2e^{|b|^2}\sum_{j=0}^\kappa\psi_j(\overline{az+b})(au+b)^j e^{|a|^2u\overline{z}+(a\overline{b}+b)u+(b\overline{a}+\overline{b})\overline{z}}\\
&=&\sum_{j=0}^\kappa\psi_j(\overline{az+b})(au+b)^j e^{u\overline{z}}\quad\text{(by \eqref{abc-cond})}.
\end{eqnarray*}
\end{proof}

\subsection{Maximality}
In this section, we show that a maximal differential operator cannot be extended as an operator in $\calF^2(\C)$ generated by the expression $\calD[\psi]$.
\begin{prop}\label{maximal-diff-op}
Let $V$, $Q$ be differential operators, induced by the families $(\phi_j)_{j=0}^\kappa$, $(\varphi)_{j=0}^\ell$, respectively (note that $V\preceq V_{\max}$, $Q\preceq Q_{\max}$). Suppose that the operator $V$ is densely defined. If $V\preceq Q$, then
$$\kappa=\ell,\quad\phi_j\equiv\varphi_j,\quad\forall j\in \{0,...,\kappa\},$$
and moreover, $Q\preceq V_{\max}$.
\end{prop}
\begin{proof}
Since $V\preceq Q$, we have $Q^*\preceq V^*$. Note that Proposition \ref{T*K_z} shows that kernel functions always belong to the domains of $Q^*$ and $V^*$. Thus, we have
$$
Q^*K_z=V^*K_z,\quad\forall z\in\C,
$$
which imply, by \eqref{T*K_z}, that
$$
\sum_{j=0}^\ell\overline{\varphi_j(z)}K_z^{[j]}=\sum_{j=0}^\kappa\overline{\phi_j(z)}K_z^{[j]},\quad\forall z\in\C.
$$
Since the set of kernel functions is linearly independent, we get the first conclusion.

The second conclusion follows from the first one and the fact that $Qf=\calD[\psi]f\in\calF^2(\C)$ for all $f\in\text{dom}(Q)$.
\end{proof}

\begin{cor}
Let $V$, $Q$ be differential operators, induced by the families $(\phi_j)_{j=0}^\kappa$, $(\varphi)_{j=0}^\ell$, respectively (note that $V\preceq V_{\max}$, $Q\preceq Q_{\max}$). Suppose that the operator $V$ is densely defined. Then $V=Q$ if and only if
$$\kappa=\ell,\quad\phi_j\equiv\varphi_j,\quad\forall j\in \{0,...,\kappa\},\quad\text{dom}(V)=\text{dom}(Q).$$
\end{cor}

\subsection{Dense domain and closed graph}
The following well known estimate asserts that convergence in the norm of $\calF^2(\C)$ implies point-wise convergence of all derivatives.
\begin{lem}
For every $f\in\calF^2(\C)$, we have
\begin{equation}\label{f(j)}
|f^{(k)}(z)|\leq e^{k(k+1)}(1+|z|)^k e^{\frac{|z|^2}{2}}\|f\|,\quad\forall z\in\C,\forall k\in\N.
\end{equation}
\end{lem}
\begin{proof}
For $k=0$ the statement follows from \cite[Corollary 2.8]{KZ}. Suppose that inequality \eqref{f(j)} holds for $k=m$. We prove the assertion for $k=m+1$.

On one hand, for $|z|\leq 1$,  by the classical Cauchy formula,
\begin{eqnarray*}
|f^{(m+1)}(z)| &\leq& \dfrac{1}{2\pi}\int_{|z-\xi|=1}\dfrac{|f^{(m)}(\xi)|}{|\xi-z|^2}|d\xi|\leq \max_{|z-\xi|=1}|f^{(m)}(\xi)|\\
&\leq& \max_{|z-\xi|=1}e^{m(m+1)}(1+|\xi|)^m e^{\frac{|\xi|^2}{2}}\|f\|\\
&\leq& e^{(m+1)(m+2)}\|f\|\quad\text{(since $|\xi|\leq |\xi-z|+|z|\leq 2$)}.
\end{eqnarray*}
On the other hand, for $|z|>1$, also by the classical Cauchy formula,
\begin{eqnarray*}
|f^{(m+1)}(z)| &\leq& |z|\max_{|z-\xi|=|z|^{-1}}|f^{(m)}(\xi)|\leq |z|\max_{|z-\xi|=|z|^{-1}}e^{m(m+1)} (1+|\xi|)^m e^{\frac{|\xi|^2}{2}}\|f\|\\
&\leq& |z| e^{m(m+1)} (1+|z|^{-1}+|z|)^m e^{\frac{(|z|^{-1}+|z|)^2}{2}}\|f\|\\
&\leq& e^{(m+1)(m+2)}(1+|z|)^{m+1} e^{\frac{|z|^2}{2}}\|f\|.
\end{eqnarray*}
Note that the third inequality holds by the inductive assumption and inequality
$$|\xi|\leq |\xi-z|+|z|,$$
while the last inequality holds, since
$$
1+|z|^{-1}+|z|\leq 2+|z|\leq 2(1+|z|)\leq e(1+|z|),
$$
and
$$
e^{\frac{(|z|^{-1}+|z|)^2}{2}}= e^{\frac{|z|^2+|z|^{-2}+2}{2}}\leq e^{\frac{|z|^2+3}{2}}\leq e^2\cdot e^{\frac{|z|^2}{2}}.
$$
\end{proof}

Recall that the maximal differential operator $T_{\max}$, induced by the family $(\psi_j)_{j=0}^\kappa$ of entire functions, is defined by
$$
T_{\max}f=\sum_{j=0}^\kappa\psi_jf^{(j)},\quad f\in\text{dom}(T_{\max}),
$$
with the maximal domain
$$
\text{dom}(T_{\max})=\{f\in\calF^2(\C):\sum_{j=0}^\kappa\psi_jf^{(j)}\in\calF^2(\C)\}.
$$
\begin{prop}\label{T-closed}
The operator $T_{\max}$ is always closed.
\end{prop}
\begin{proof}
Let $(f_n)$ be a sequence of functions in $\text{dom}(T_{\max})$ and $f,g\in\calF^2(\C)$, such that
$$
f_n\to f \quad\text{and}\quad T_{\max}f_n\to g\quad\hbox{in $\calF^2(\C)$}.
$$
Since convergence in the norm of $\calF^2(\C)$ implies a point convergence of derivatives, we have
$$
f_n^{(j)}(z)\to f^{(j)}(z) \quad\text{and}\quad T_{\max}f_n(z)\to g(z),\quad\forall z\in\C.
$$
On the other hand,
$$
T_{\max}f_n(z)=\sum_{j=0}^\kappa\psi_j(z)f_n^{(j)}(z),\quad\forall z\in\C.
$$
Therefore,
$$
\sum_{j=0}^\kappa\psi_j(z)f^{(j)}(z)=g(z),\quad\forall z\in \C,\ \text{which means}\ \sum_{j=0}^\kappa\psi_jf^{(j)}=g.
$$
Since $g\in\calF^2(\C)$, we derive that $f\in\text{dom}(T_{\max})$ and $T_{\max}f=g$.
\end{proof}

Denote by $\calK$ the algebraic linear span of $\{K_z:z\in\C\}$, regarded as a dense subspace of Fock space. The following result provides an alternate description of the maximal differential operators.
\begin{prop}
Let $\calQ$ be the linear operator defined by
$$
\text{dom}(\calQ)=\calK,\quad\calQ K_z=\sum_{j=0}^{\kappa}\overline{\psi_j(z)}K_z^{[j]},
$$
Then $T_{\max}=\calQ^*$. Moreover, $T_{\max}$ is densely defined if and only if $\calQ$ is closable.
\end{prop}
\begin{proof}
Take arbitrarily $f=\sum_{j=1}^n\lambda_j K_{z_j}\in\text{dom}(\calQ)$. For every $g\in\calF^2(\C)$, we have
\begin{eqnarray*}
\langle \calQ f,g\rangle &=& \sum_{j=1}^n\lambda_j\langle \calQ K_{z_j},g\rangle=\sum_{j=1}^n\sum_{\ell=0}^\kappa\lambda_j\overline{\psi_\ell(z_j)g^{(\ell)}(z_j)}=\sum_{j=1}^n\lambda_j\overline{\calD[\psi] g(z_j)}.
\end{eqnarray*}
Note that by Riesz' lemma, the function $g$ belongs to $\text{dom}(\calQ^*)$ if and only if there exists $C>0$ such that
$$
|\langle \calQ f,g\rangle|\leq C\|f\|,\quad\forall f\in\text{dom}(\calQ),
$$
or equivalently, if and only if
$$
|\sum_{j=1}^n\calD[\psi] g(z_j)\overline{\lambda_j}|^2\leq C^2\sum_{j,\ell=1}^n \lambda_j\overline{\lambda_\ell}K_{z_j}(z_\ell).
$$
In view of \cite{FHS}, the latter is equivalent to $\calD[\psi] g\in\calF^2(\C)$. This shows that $\text{dom}(\calQ^*)=\text{dom}(T_{\max})$. Moreover,
$$
\langle \calQ f,g\rangle=\langle f,\calD[\psi] g\rangle=\langle f,T_{\max}g\rangle,\quad\forall f\in\text{dom}(\calQ),\forall g\in\text{dom}(T_{\max}),
$$
which gives $T_{\max}=\calQ^*$.

The rest conclusion is a direct consequence of the von Neumann theorem, see for instance \cite{Akhiezer-Glazman}.
\end{proof}

\subsection{Adjoint}
As it will be seen in full detail in the next section, for $\calC$-selfadjoint maximal differential operators, the symbol $\psi_p$ has to be a polynomial form. Thus, we pause for a while and explore an explicit form for the adjoint $T^*$ on $\calF^2(\C)$, in the case of a polynomial total symbol.
\begin{thm}\label{adjoint-form}
Let $T_{\max}$ be a maximal differential operator of order $\kappa$, induced by the symbols
$$
\psi_p(z)=\sum_{j=0}^\kappa d_{j,p}(az+b)^j,\quad\forall p\in\{0,...,\kappa\},
$$
$a,b,c$ be constants satisfying \eqref{abc-cond}. Then $T_{\max}^*=S_{\max}$, where $S_{\max}$ is the maximal differential operator induced by the symbols
$$
\widehat{\psi_p}(z)=\sum_{j=0}^\kappa z^j\sum_{\ell=p}^\kappa \binom{\ell}{p}\overline{d_{\ell,j} a^p b^{\ell-p}},\quad\forall p\in\{0,...,\kappa\}.
$$
\end{thm}
\begin{proof}
The proof is separated into two steps.

{\bf\noindent Step 1:} We claim that $T_{\max}^*\preceq S_{\max}$.

Let  $z,u\in\C$ and $f\in\text{dom}(T_{\max}^*)$. On one hand, we have
\begin{eqnarray*}
T_{\max}K_z(u) 
&=& \sum_{j=0}^\kappa \sum_{\ell=0}^\kappa d_{\ell,j}(au+b)^\ell\overline{z}^j K_z(u)\\
&=& \sum_{j=0}^\kappa \sum_{\ell=0}^\kappa \sum_{q=0}^\ell d_{\ell,j}\binom{\ell}{q}a^q u^q b^{\ell-q}\overline{z}^j K_z(u)\\
&=& \sum_{j=0}^\kappa \sum_{\ell=0}^\kappa \sum_{q=0}^\ell \binom{\ell}{q}d_{\ell,j} a^q b^{\ell-q}\overline{z}^j K_z^{[q]}(u),
\end{eqnarray*}
which implies, taking inner product both sides with $f$, that
\begin{eqnarray*}
\langle f,T_{\max}K_z\rangle
&=& \sum_{j=0}^\kappa \sum_{\ell=0}^\kappa \sum_{q=0}^\ell \binom{\ell}{q}z^j\overline{d_{\ell,j} a^q b^{\ell-q}} f^{(q)}(z)\\
&=& \sum_{q=0}^\kappa\sum_{j=0}^\kappa\sum_{\ell=q}^\kappa \binom{\ell}{q}z^j\overline{d_{\ell,j} a^q b^{\ell-q}} f^{(q)}(z).
\end{eqnarray*}
On the other hand, by the definitions of adjoint and kernel functions,
$$
\langle f,T_{\max}K_z\rangle= \langle T_{\max}^*f,K_z\rangle=T_{\max}^*f(z).
$$
These show that
$$
T_{\max}^*f(z)=\sum_{q=0}^\kappa f^{(q)}(z)\sum_{j=0}^\kappa z^j\sum_{\ell=q}^\kappa \binom{\ell}{q}\overline{d_{\ell,j} a^q b^{\ell-q}}=S_{\max}f(z).
$$
Since $f$ was arbitrary in $\text{dom}(T_{\max}^*)$, we get the first inclusion.
\bigskip

{\bf\noindent Step 2:} The converse inclusion holds, i.e. $S_{\max}\preceq T_{\max}^*$. 

It is sufficient to prove that
$$
\langle T_{\max}f,g\rangle=\langle f,S_{\max}g\rangle,\quad\forall f\in\text{dom}(T_{\max}),\forall g\in\text{dom}(S_{\max}).
$$
Let $f\in\text{dom}(T_{\max})$ and $g\in\text{dom}(S_{\max})$.

A direct computation yields
$$
T_{\max}f(z)=\sum_{j=0}^\kappa \psi_j(z)f^{(j)}(z)=\sum_{j=0}^\kappa \sum_{\ell=0}^\kappa \sum_{p=0}^\ell\binom{\ell}{p}d_{\ell,j}a^p b^{\ell-p} z^p f^{(j)}(z),
$$
and
$$
S_{\max}g(z) = \sum_{t=0}^\kappa\widehat{\psi_t}(z)g^{(t)}(z)=\sum_{t=0}^\kappa\sum_{s=0}^\kappa \sum_{r=t}^\kappa \binom{r}{t}\overline{d_{r,s} a^t b^{r-t}}z^s g^{(t)}(z),
$$
where
$$
\psi_j(z)=\sum_{\ell=0}^\kappa d_{\ell,j}(az+b)^\ell=\sum_{\ell=0}^\kappa d_{\ell,j}\sum_{p=0}^\ell\binom{\ell}{p}a^p z^p b^{\ell-p},
$$
and
$$
\widehat{\psi_t}(z)=\sum_{s=0}^\kappa z^s\sum_{r=t}^\kappa \binom{r}{t}\overline{d_{r,s} a^t b^{r-t}}.
$$
Note that for every entire function $h(z)=\sum_{n\geq 0}h_n z^n$ ($h_n$ are the Taylor coefficients), its $m$th-order derivative is
$$
h^{(m)}(z)=\sum_{n\geq m}n(n-1)...(n-m+1)h_n z^{n-m}.
$$
Hence
$$
z^p f^{(j)}(z)=\sum_{n\geq j}n(n-1)...(n-j+1)f_n z^{n-j+p},
$$
while
$$
z^s g^{(t)}(z)=\sum_{\alpha\geq t}\alpha(\alpha-1)...(\alpha-t+1)g_\alpha z^{\alpha-t+s}.
$$
Here, $f_n$, $g_\alpha$ are Taylor coefficients of $f,g$, respectively.

According to formula \eqref{inner-product}:
\begin{eqnarray*}
&&\langle T_{\max}f,g\rangle \\
&&= \sum_{j=0}^\kappa \sum_{\ell=0}^\kappa \sum_{p=0}^\ell\sum_{n\geq j} \binom{\ell}{p}d_{\ell,j}a^p b^{\ell-p}n(n-1)...(n-j+1)f_n\overline{g_{n+p-j}}\|z^{n+p-j}\|^2,
\end{eqnarray*}
and
\begin{eqnarray*}
&&\langle f,S_{\max}g\rangle \\
&&=\sum_{t=0}^\kappa\sum_{s=0}^\kappa \sum_{r=t}^\kappa \sum_{\alpha\geq t}\binom{r}{t}d_{r,s} a^t b^{r-t}\alpha(\alpha-1)...(\alpha-t+1)f_{s+\alpha-t}\overline{g_\alpha}\|z^{s+\alpha-t}\|^2.
\end{eqnarray*}
Changing the variables $s=j$, $t=p$, $r=\ell$ in the last identity implies
$$
\langle f,S_{\max}g\rangle =\sum_{p=0}^\kappa\sum_{j=0}^\kappa \sum_{\ell=p}^\kappa \sum_{\alpha\geq p}\binom{\ell}{p}d_{\ell,j} a^p b^{\ell-p}\alpha(\alpha-1)...(\alpha-p+1)f_{j+\alpha-p}\overline{g_\alpha}\|z^{j+\alpha-p}\|^2.
$$
With $n=j+\alpha-p$ and the observation
$$
\alpha\geq p\Longleftrightarrow n\geq j
$$
gives
\begin{eqnarray*}
&&\sum_{\alpha\geq p}\binom{\ell}{p}d_{\ell,j} a^p b^{\ell-p}\alpha(\alpha-1)...(\alpha-p+1)f_{j+\alpha-p}\overline{g_\alpha}\|z^{j+\alpha-p}\|^2\\
&&=\sum_{n\geq j}\binom{\ell}{p}d_{\ell,j} a^p b^{\ell-p}(n-j+p)(n-j+p-1)...(n-j+1)f_n\overline{g_{n-j+p}}\|z^n\|^2\\
&&=\sum_{n\geq j}\binom{\ell}{p}d_{\ell,j} a^p b^{\ell-p}n(n-1)...(n-j+1)f_n\overline{g_{n-j+p}}\|z^{n+p-j}\|^2.
\end{eqnarray*}
Note that the last identity holds, since $\|z^n\|^2=n!$. Taking the sum over $\ell$, $j$ and $p$, we get
\begin{eqnarray*}
&&\langle f,S_{\max}g\rangle\\
&&=\sum_{p=0}^\kappa\sum_{j=0}^\kappa\sum_{\ell=p}^\kappa \sum_{n\geq j}\binom{\ell}{p}d_{\ell,j} a^p b^{\ell-p}n(n-1)...(n-j+1)f_n\overline{g_{n-j+p}}\|z^{n+p-j}\|^2\\
&&= \sum_{j=0}^\kappa \sum_{\ell=0}^\kappa \sum_{p=0}^\ell\sum_{n\geq j} \binom{\ell}{p}d_{\ell,j}a^p b^{\ell-p}n(n-1)...(n-j+1)f_n\overline{g_{n+p-j}}\|z^{n+p-j}\|^2.
\end{eqnarray*}
\end{proof}

\section{Complex symmetry with maximal domain}\label{c-self-self}
The present section is devoted to a detailed study of complex-selfadjoint operators with maximal domain of definition, with respect to the three parameter family of natural
conjugations acting on Fock space.

 The maximal linear differential operator $T_{\max}$, associated to the coefficient data $(\psi_j)_{j=0}^\kappa$ consisting of entire functions, is defined by
$$
T_{\max}f=\sum_{j=0}^\kappa\psi_jf^{(j)},\quad f\in\text{dom}(T_{\max}),
$$
with (maximal) domain
$$
\text{dom}(T_{\max})=\{f\in\calF^2(\C):\sum_{j=0}^\kappa\psi_jf^{(j)}\in\calF^2(\C)\}.
$$
\subsection{$\calC_{a,b,c}$-selfadjointness}\label{c-self}

The following lemma is instrumental in identifying the polynomial coefficients.
\begin{lem}\label{psi-form}
Suppose that the family $(\psi_j)_{j=0}^\kappa$ of entire functions satisfies
\begin{equation}\label{important}
\sum_{j=0}^\kappa\psi_j(u)\overline{z}^j=\sum_{j=0}^\kappa\psi_j(\overline{az+b})(au+b)^j,\quad\forall z,u\in\C.
\end{equation}
Then these functions are polynomials of the forms:
\begin{equation}\label{form-psi}
\psi_p(u)=\sum_{j=0}^\kappa d_{j,p}(au+b)^j,\quad \forall p\in\{0,...,\kappa\},
\end{equation}
where $d_{j,p}$ are constants. 

Moreover, the identity
\begin{equation}\label{form-psi-induc}
\sum_{j=p+1}^\kappa\psi_j(u)\overline{z}^{j-p-1}=\sum_{j=0}^\kappa\left(\dfrac{\psi_j(\overline{az+b})-\calQ_{j,p}(z)}{\overline{z}^{p+1}}\right)(au+b)^j,
\end{equation}
holds, where
$$
\calQ_{j,p}(z)=\sum_{\ell=0}^p\dfrac{\psi_j^{(\ell)}(\overline{b})(\overline{az})^\ell}{\ell !}.
$$
If the parameters $a,b$ satisfy \eqref{abc-cond}, then $d_{j,p}=d_{p,j}$.
\end{lem}
\begin{proof}
We prove this lemma by induction on $p$. Here, $$d_{j,p}=\frac{\overline{a}^p\psi_j^{(p)}(\overline{b})}{p!}.$$

Indeed, letting $z=0$ in \eqref{important} gives
$$
\psi_0(u)=\sum_{j=0}^\kappa\psi_j(\overline{b})(au+b)^j.
$$
Substitute $\psi_0$ back into \eqref{important} to obtain
$$
\sum_{j=0}^\kappa\psi_j(\overline{az+b})(au+b)^j=\sum_{j=0}^\kappa\psi_j(\overline{b})(au+b)^j+\sum_{j=1}^\kappa\psi_j(u)\overline{z}^j,
$$
which gives
\begin{equation*}\label{important-1}
\sum_{j=1}^\kappa\psi_j(u)\overline{z}^{j-1}=\sum_{j=0}^\kappa\left(\dfrac{\psi_j(\overline{az+b})-\psi_j(\overline{b})}{\overline{z}}\right)(au+b)^j.
\end{equation*}
These shows that \eqref{form-psi}-\eqref{form-psi-induc} hold for $p=0$. Suppose that they hold for $p=m$. We prove for $p=m+1$.

Identity \eqref{form-psi-induc} (when $p=m$) is rewritten as
\begin{eqnarray*}
\psi_{m+1}(u)&=&\sum_{j=0}^\kappa\left(\dfrac{\psi_j(\overline{az+b})-\calQ_{j,m}(z)}{\overline{z}^{m+1}}\right)(au+b)^j-\sum_{j=m+2}^\kappa\psi_j(u)\overline{z}^{j-m-1}.
\end{eqnarray*}
Letting $z\to 0$ in the above identity gives the explicit form of $\psi_{m+1}$.

Also by \eqref{form-psi-induc} (when $p=m$), we see
\begin{eqnarray*}
\sum_{j=m+2}^\kappa\psi_j(u)\overline{z}^{j-m-1}&=&\sum_{j=m+1}^\kappa\psi_j(u)\overline{z}^{j-m-1}-\psi_{m+1}(u)\\
&=&\sum_{j=0}^\kappa\left(\dfrac{\psi_j(\overline{az+b})-\calQ_{j,m+1}(z)}{\overline{z}^{m+1}}\right)(au+b)^j.
\end{eqnarray*}
Dividing both sides by $\overline{z}$, we get \eqref{form-psi-induc} with $p=m+1$.

Suppose that $a,b$ satisfy \eqref{abc-cond}. Substituting \eqref{form-psi} back into \eqref{important}, we get
\begin{eqnarray*}
\sum_{j=0}^\kappa\sum_{t=0}^\kappa d_{t,j}(au+b)^t\overline{z}^j &=& \sum_{j=0}^\kappa\sum_{t=0}^\kappa d_{t,j}(|a|^2\overline{z}+a\overline{b}+b)^t (au+b)^j\\
&=& \sum_{j=0}^\kappa\sum_{t=0}^\kappa d_{t,j}\overline{z}^t (au+b)^j\quad\text{(by \eqref{abc-cond})},
\end{eqnarray*}
which implies $d_{t,j}=d_{j,t}$.
\end{proof}

The following result isolates a necessary condition for maximal differential operators to be $\calC$-selfadjoint with respect to the conjugation $\calC_{a,b,c}$.
\begin{prop}\label{=>}
Let $T_{\max}$ be a maximal differential operator of order $\kappa$, induced by $(\psi_j)_{j=0}^\kappa$, and $\calC_{a,b,c}$ a weighted composition conjugation. If 
$$
T_{\max}K_z=\calC_{a,b,c}T_{\max}^*\calC_{a,b,c}K_z,\quad\forall z\in\C,
$$
then the symbols are of the following forms
$$
\psi_p(z)=\sum_{j=0}^\kappa d_{j,p}(az+b)^j,\quad \forall p\in\{0,...,\kappa\},
$$
where $d_{j,p}$ are constant numbers satisfying $d_{j,p}=d_{p,j}$.
\end{prop}
\begin{proof}
A direct computation yields
$$
T_{\max}K_z(u)=\sum_{j=0}^\kappa\psi_j(u)\overline{z}^j e^{u\overline{z}}.
$$
On the other hand, by assumption and Proposition \ref{CT*CKz}(2),
$$
T_{\max}K_z(u)=\calC_{a,b,c}T_{\max}^*\calC_{a,b,c}K_z(u)=\sum_{j=0}^\kappa\psi_j(\overline{az+b})(au+b)^j e^{u\overline{z}}.
$$

Thus, the identity $T_{\max}K_z=\calC_{a,b,c}T_{\max}^*\calC_{a,b,c}K_z$ is reduced to \eqref{important}, and hence, by Lemma \ref{psi-form}, we get the explicit forms of $\psi_p$.
\end{proof}

It turns out that the condition in Proposition \ref{=>} is also sufficient for a maximal differential operator to be $\calC$-selfadjoint. As a first step in proving the converse implication 
we state the following algebraic observation.

\begin{prop}\label{non-max}
Let $\calC_{a,b,c}$ be a weighted composition conjugation. Furthermore, let $\calD[\psi]$, $\calD[\widehat{\psi}]$ be the linear differential expressions induced by
$$
\psi_p(z)=\sum_{j=0}^\kappa d_{j,p}(az+b)^j,\quad\forall p\in\{0,...,\kappa\},
$$
$$
\widehat{\psi_p}(z)=\sum_{j=0}^\kappa z^j\sum_{\ell=p}^\kappa \binom{\ell}{p}\overline{d_{\ell,j} a^p b^{\ell-p}},\quad\forall p\in\{0,...,\kappa\},
$$
where $d_{j,p}=d_{p,j}$, respectively. Then $\calC_{a,b,c}\calD[\psi]\calC_{a,b,c}=\calD[\widehat{\psi}]$.
\end{prop}
\begin{proof}
Note that by an induction argument, we infer
$$
(\calC_{a,b,c}h)^{(p)}(z)=ce^{bz}\left[\sum_{j=0}^p \binom{p}{j}a^j b^{p-j}\overline{h^{(j)}(\overline{az+b})}\right],\quad\forall h\in\calF^2(\C).
$$
We have
\begin{eqnarray*}
(\calD[\psi]\calC_{a,b,c}f)(z)&=&\sum_{p=0}^\kappa\psi_p(z)ce^{bz}\left[\sum_{j=0}^p \binom{p}{j}a^j b^{p-j}\overline{f^{(j)}(\overline{az+b})}\right]\\
&=&\sum_{p=0}^\kappa\sum_{j=0}^p\sum_{\ell=0}^\kappa \binom{p}{j}cd_{\ell,p}a^j b^{p-j}e^{bz}(az+b)^\ell\overline{f^{(j)}(\overline{az+b})},
\end{eqnarray*}
and hence,
\begin{eqnarray*}
&& (\calD[\psi]\calC_{a,b,c})f(\overline{az+b})\\
&&=\sum_{p=0}^\kappa\sum_{j=0}^p\sum_{\ell=0}^\kappa \binom{p}{j}cd_{\ell,p}a^j b^{p-j}e^{b\overline{az}+|b|^2}(|a|^2\overline{z}+a\overline{b}+b)^\ell\overline{f^{(j)}(|a|^2z+\overline{a}b+\overline{b})}\\
&&=\sum_{p=0}^\kappa\sum_{j=0}^p\sum_{\ell=0}^\kappa \binom{p}{j}ce^{|b|^2}d_{\ell,p}a^j b^{p-j}\overline{e^{a\overline{b}z}z^\ell f^{(j)}(z)}\quad\text{(by \eqref{abc-cond})}.
\end{eqnarray*}
Acting $\calC_{a,b,c}$ on $\calD[\psi]\calC_{a,b,c}f$ gives
\begin{eqnarray*}
(\calC_{a,b,c}\calD[\psi]\calC_{a,b,c}f)(z)&=& ce^{bz}\overline{(\calD[\psi]\calC_{a,b,c}f)(\overline{az+b})}\\
&=&\sum_{p=0}^\kappa\sum_{j=0}^p\sum_{\ell=0}^\kappa \binom{p}{j}|c|^2 e^{|b|^2}\overline{d_{\ell,p}a^j b^{p-j}}e^{(a\overline{b}+b)z}z^\ell f^{(j)}(z)\\
&=&\sum_{p=0}^\kappa\sum_{j=0}^p\sum_{\ell=0}^\kappa \binom{p}{j}\overline{d_{\ell,p}a^j b^{p-j}}z^\ell f^{(j)}(z)\quad\text{(by \eqref{abc-cond})}.
\end{eqnarray*}
Interchanging the order of summation, we get
\begin{eqnarray*}
(\calC_{a,b,c}\calD[\psi]\calC_{a,b,c}f)(z)&=& \sum_{j=0}^\kappa\sum_{\ell=0}^\kappa\sum_{p=j}^\kappa \binom{p}{j}\overline{d_{\ell,p}a^j b^{p-j}}z^\ell f^{(j)}(z)\\
&=& \sum_{j=0}^\kappa \widehat{\psi_j}(z)f^{(j)}(z)\quad\text{(since $d_{\ell,p}=d_{p,\ell}$)}\\
&=& (\calD[\widehat{\psi}]f)(z).
\end{eqnarray*}
\end{proof}

With all preparation in place we can now state the main result of this section.
\begin{thm}\label{Cself}
Let $\calC_{a,b,c}$ be a weighted composition conjugation, and $T_{\max}$ a maximal differential operator of order $\kappa$, induced by the family $(\psi_j)_{j=0}^\kappa$ of entire functions. The following assertions are equivalent.
\begin{enumerate}
\item The operator $T_{\max}$ is $\calC_{a,b,c}$-selfadjoint.
\item The operator $T_{\max}$ is densely defined and satisfies $T_{
\max}^*\preceq\calC_{a,b,c}T\calC_{a,b,c}$.
\item The symbols are of the following forms
\begin{equation}\label{psi_p-form}
\psi_p(z)=\sum_{j=0}^\kappa d_{j,p}(az+b)^j,\quad \forall p\in\{0,...,\kappa\},
\end{equation}
where $d_{j,p}$ are constant numbers satisfying $d_{j,p}=d_{p,j}$.
\end{enumerate}
\end{thm}
\begin{proof}
It is clear that $(1)\Longrightarrow(2)$, while implication $(2)\Longrightarrow(3)$ holds by Proposition \ref{=>}.

It remains to verify the implication $(3)\Longrightarrow(1)$. Indeed, suppose that assertion (3) holds. Note that by Proposition \ref{T-closed}, the operator $T_{\max}$ is always closed. We separate the proof into two steps.
\bigskip

{\bf\noindent Step 1:} The operator $T_{\max}$ is densely defined. More specifically, we show that $K_z\in\text{dom}(T_{\max})$ for every $z\in\C$. 

Indeed, 
\begin{eqnarray*}
\sum_{p=0}^\kappa (\psi_p K_z^{(p)})(u)&=&\sum_{p=0}^\kappa \psi_p(u)\overline{z}^p e^{\overline{z}u}=\sum_{p=0}^\kappa \sum_{j=0}^\kappa d_{j,p}(au+b)^j\overline{z}^p e^{\overline{z}u}\\
&=&\sum_{p=0}^\kappa \sum_{j=0}^\kappa d_{j,p}\sum_{\ell=0}^j\binom{j}{\ell}a^\ell u^\ell b^{j-\ell}\overline{z}^p e^{\overline{z}u}\\
&=&\sum_{p=0}^\kappa \sum_{j=0}^\kappa \sum_{\ell=0}^j\binom{j}{\ell}d_{j,p}a^\ell b^{j-\ell}\overline{z}^p K_z^{[\ell]}(u),
\end{eqnarray*}
which implies that $\sum_{p=0}^\kappa \psi_p K_z^{(p)}$ is a linear combination of elements $K_z^{[\ell]}$, and hence, it must be in $\text{dom}(T_{\max})$.
\bigskip

{\bf\noindent Step 2:} The identity $\calC_{a,b,c}T_{\max}\calC_{a,b,c}=T_{\max}^*$ holds true. 

By Proposition \ref{adjoint-form}, it is enough to show that
$$
\calC_{a,b,c}T_{\max}\calC_{a,b,c}=S_{\max},
$$
where $S_{\max}$ is the maximal differential operator induced by
$$
\widehat{\psi_p}(z)=\sum_{j=0}^\kappa z^j\sum_{\ell=p}^\kappa \binom{\ell}{p}\overline{d_{\ell,j} a^p b^{\ell-p}},\quad \forall p\in\{0,...,\kappa\}.
$$
{\it\noindent Step 2.1:} We show that $\calC_{a,b,c}T_{\max}\calC_{a,b,c}\preceq S_{\max}$.

By Proposition \ref{non-max}, the differential operator expression of $\calC_{a,b,c}T_{\max}\calC_{a,b,c}$ is induced by the family $(\widehat{\psi_j})_{j=0}^\kappa$. Since $\text{dom}(S_{\max})$ is maximal, $\text{dom}(\calC_{a,b,c}T_{\max}\calC_{a,b,c})$ is a subspace of $\text{dom}(S_{\max})$.

{\it\noindent Step 2.2:} To complete the proof, we verify that 
$$\text{dom}(S_{\max})\subseteq \text{dom}(\calC_{a,b,c}T_{\max}\calC_{a,b,c}).$$

Let $g\in\text{dom}(S_{\max})$. Again Proposition \ref{non-max} implies
$$
S_{\max}g(z)=\sum_{j=0}^\kappa\widehat{\psi_j}(z)g^{(j)}(z)=\calC_{a,b,c}\calD[\psi]\calC_{a,b,c}g(z),
$$
whence
$$
\calC_{a,b,c}S_{\max}g(z)=\sum_{p=0}^\kappa\psi_p(z) (\calC_{a,b,c}f)^{(p)}(z).
$$
Since $S_{\max}g\in\calF^2(\C)$ and $\calC_{a,b,c}$ is a conjugation on $\calF^2(\C)$, $\calC_{a,b,c}S_{\max}g\in\calF^2(\C)$, and hence, $\sum_{p=0}^\kappa\psi_p(\calC_{a,b,c}g)^{(p)}\in\calF^2(\C)$. This implies $g\in\text{dom}(\calC_{a,b,c}T_{\max}\calC_{a,b,c})$.
\end{proof}

\subsection{Selfadjointness}\label{self}
In complete analogy to the studied $\calC$-symmetry, we give below a description of maximal linear differential operators which are selfadjoint on Fock space $\calF^2(\C)$.

As in the complex symmetry scenario, we formulate first an algebraic lemma.
\begin{lem}\label{psi-form-ss}
Suppose that the family $(\psi_j)_{j=0}^\kappa$ of entire functions satisfies
\begin{equation}\label{important-s}
\sum_{j=0}^\kappa\psi_j(u)\overline{z}^j=\sum_{j=0}^\kappa\overline{\psi_j(z)}u^j,\quad\forall z,u\in\C.
\end{equation}
Then these functions are polynomials of the forms
\begin{equation}\label{form-psi-s}
\psi_p(u)=\sum_{j=0}^\kappa d_{j,p}u^j,\quad \forall p\in\{0,...,\kappa\},
\end{equation}
where $d_{j,p}$ are constant numbers satisfying $d_{j,p}=\overline{d_{p,j}}$.

Moreover, the identity
\begin{equation}\label{form-psi-induc-s}
\sum_{j=p+1}^\kappa\psi_j(u)\overline{z}^{j-p-1}=\sum_{j=0}^\kappa\left(\dfrac{\overline{\psi_j(z)}-\overline{Q_{j,p}(z)}}{\overline{z}^{p+1}}\right)u^j,
\end{equation}
holds, where
$$
Q_{j,p}(z)=\sum_{\ell=0}^p\dfrac{\psi_j^{(\ell)}(0)z^\ell}{\ell !}.
$$
\end{lem}
\begin{proof}
We prove this lemma by induction on $p$. Here,
$$d_{j,p}=\frac{\overline{\psi_j^{(p)}(0)}}{p!}.$$

Indeed, letting $z=0$ in \eqref{important-s} gives
$$
\psi_0(u)=\sum_{j=0}^\kappa\overline{\psi_j(0)}u^j.
$$
Substitute $\psi_0$ back into \eqref{important-s} to obtain
$$
\sum_{j=0}^\kappa\overline{\psi_j(0)}u^j+\sum_{j=1}^\kappa\psi_j(u)\overline{z}^j=\sum_{j=0}^\kappa\overline{\psi_j(z)}u^j,
$$
which implies that
\begin{equation*}\label{important-1}
\sum_{j=1}^\kappa\psi_j(u)\overline{z}^{j-1}=\sum_{j=0}^\kappa\left(\dfrac{\overline{\psi_j(z)}-\overline{\psi_j(0)}}{\overline{z}}\right)u^j.
\end{equation*}
This shows that the conclusions hold for $p=0$. Suppose that they hold for $p=m$.

So, we can rewrite \eqref{form-psi-induc-s} (when $p=m$) as follows
$$
\psi_{m+1}(u)=\sum_{j=0}^\kappa\left(\dfrac{\overline{\psi_j(z)}-\overline{Q_{j,m}(z)}}{\overline{z}^{m+1}}\right)u^j-\sum_{j=m+2}^\kappa\psi_j(u)\overline{z}^{j-m-1}.
$$
Letting $z\to 0$ in the above identity gives the explicit form of $\psi_{m+1}$. 

Substituting this form back into \eqref{form-psi-induc-s} (when $p=m$), we get
\begin{eqnarray*}
\sum_{j=m+2}^\kappa\psi_j(u)\overline{z}^{j-m-1}&=&\sum_{j=m+1}^\kappa\psi_j(u)\overline{z}^{j-m-1}-\psi_{m+1}(u)\\
&=&\sum_{j=0}^\kappa\left(\dfrac{\overline{\psi_j(z)}-\overline{Q_{j,m+1}(z)}}{\overline{z}^{m+1}}\right)u^j.
\end{eqnarray*}
Dividing both sides by $\overline{z}$, we obtain \eqref{form-psi-induc-s}, with $p=m+1$. Substituting \eqref{form-psi-s} back into \eqref{important-s}, we get
$$
\sum_{j=0}^\kappa\sum_{t=0}^\kappa d_{t,j}u^t\overline{z}^j = \sum_{t=0}^\kappa\sum_{j=0}^\kappa \overline{d_{t,j}}\overline{z}^t u^j=\sum_{j=0}^\kappa\sum_{t=0}^\kappa \overline{d_{j,t}}\overline{z}^j u^t,
$$
which implies that $d_{t,j}=\overline{d_{j,t}}$.
\end{proof}

A necessary condition for maximal differential operators to be selfadjoint is provided by the following proposition.
\begin{prop}\label{=>s}
Let $T_{\max}$ be a maximal differential operator induced by $(\psi_j)_{j=0}^\kappa$. If it satisfies
$$
T_{\max}K_z=T_{\max}^*K_z,\quad\forall z\in\C,
$$
then the symbols are of the following forms
$$
\psi_p(z)=\sum_{j=0}^\kappa d_{j,p}z^j,\quad \forall p\in\{0,...,\kappa\},
$$
where $d_{j,p}$ are constant numbers satisfying $d_{j,p}=\overline{d_{p,j}}$. 
\end{prop}
\begin{proof}
A direct computation shows that
$$
T_{\max}K_z(u)=\sum_{j=0}^\kappa\psi_j(u)\overline{z}^j e^{u\overline{z}}.
$$
On the other hand, by Proposition \ref{T*K_z},
$$
T_{\max}^*K_z(u)=\sum_{j=0}^\kappa\overline{\psi_j(z)}u^j e^{u\overline{z}}.
$$
Thus, identity $T_{\max}K_z=T_{\max}^*K_z$ is reduced to \eqref{important-s}, and hence by Lemma \ref{psi-form-ss}, we get the desired conclusion.
\end{proof}

It turns out that the assertion in Proposition \ref{=>s}  is also sufficient for a maximal differential operator to be selfadjoint.

\begin{thm}\label{selfadjoint}
Let $T_{\max}$ be a maximal differential operator of order $\kappa$, induced by the family $(\psi_j)_{j=0}^\kappa$ of entire functions. The following assertions are equivalent.
\begin{enumerate}
\item The operator $T_{\max}$ is selfadjoint.
\item The operator $T_{\max}$ is densely defined and satisfies $T_{\max}^*\preceq T_{\max}$.
\item The symbols are of the following forms
\begin{equation}\label{psi_p-form-ss}
\psi_p(z)=\sum_{j=0}^\kappa d_{j,p}z^j,\quad \forall p\in\{0,...,\kappa\},
\end{equation}
where $d_{j,p}$ are constant numbers satisfying $d_{j,p}=\overline{d_{p,j}}$.
\end{enumerate}
\end{thm}
\begin{proof}
It is clear that $(1)\Longrightarrow(2)$, while implication $(2)\Longrightarrow(3)$ holds by Proposition \ref{=>s}.
It remains  to show that $(3)\Longrightarrow(1)$. Indeed, suppose that assertion (3) holds. Note that by Proposition \ref{T-closed}, $T_{\max}$ is always closed. By arguments similar to those used in {\bf Step 1} of Theorem \ref{Cself}, we can show that the operator $T_{\max}$ is densely defined.

Note that in virtue of Proposition \ref{adjoint-form} (here $a=1$ and $b=0$), $T_{\max}^*=S_{\max}$, where $S_{\max}$ is the maximal differential operator induced by
$$
\widehat{\psi_p}(z)=\sum_{j=0}^\kappa z^j\sum_{\ell=p}^\kappa \binom{\ell}{p}\overline{d_{\ell,j} a^p b^{\ell-p}}=\sum_{j=0}^\kappa z^j\overline{d_{p,j}}.
$$
Since $d_{j,p}=\overline{d_{p,j}}$, we must have $\widehat{\psi_p}=\psi_p$, and hence, $S_{\max}=T_{\max}$. Finally $T_{\max}^*=S_{\max}=T_{\max}$ which proves the theorem.
\end{proof}

As mentioned in the Introduction, the class of complex symmetric operators is large enough to cover all selfadjoint operators see for instance \cite{GPP}. Naturally, one can ask how big is the $\calC_{a,b,c}$-selfadjoint class. The following result explores this question in a simple case.
\begin{prop}
Let $T_{\max}$ be the selfadjoint maximal differential operator, induced by
$$
\psi_0(z)=d_{0,0}+d_{n,0}z^n,\quad\psi_n(z)=d_{0,n}+d_{n,n}z^n,\quad\psi_j\equiv\mathbf{0},\,\forall j\in\{1,\cdots,n-1\},
$$
where $n\geq 1$,  $d_{n,0}=\overline{d_{0,n}}$ and $d_{0,0},d_{n,n}$ are real. Then $T_{\max}$  is $\calC_{a,b,c}$-selfadjoint with
$$
b=0,\quad c=1,\quad a=
\begin{cases}
1,\quad\text{if $d_{0,n}=0$},\\
d_{n,0}/d_{0,n},\quad\text{if $d_{0,n}\ne 0$}.
\end{cases}
$$
\end{prop}
\begin{proof}
With such choice, $\calC_{a,b,c}$ is a conjugation by \eqref{abc-cond}, and hence, by Theorem \ref{Cself}, $T_{\max}$ is $\calC$-selfadjoint with respect to the conjugation $\calC_{a,b,c}$.
\end{proof}

In other terms,

\begin{cor}
Every selfadjoint first-order maximal differential operator is $\calC$-sefladjoint with respect to some weighted composition conjugation.
\end{cor}

\section{Complex symmetry with arbitrary domain}\label{non-maximal-diff-op}
In the previous section we studied the complex symmetry of linear differential operators with maximal domains.  In the present section we relax the domain assumption to only discover that
$C$-selfadjointness cannot be separated from the maximality of the domain.

\begin{thm}\label{no-max-do}
Let $T$ be a $\kappa$-order differential operator induced by the family $(\psi_j)_{j=0}^\kappa$ of entire functions (note that $T\preceq T_{\max}$). Then it is $\calC$-selfadjoint with respect to the conjugation $\calC_{a,b,c}$ if and only if the following conditions hold.
\begin{enumerate}
\item $T=T_{\max}$.
\item The symbols are of forms
$$
\psi_p(z)=\sum_{j=0}^\kappa d_{j,p}(az+b)^j,\quad \forall p\in\{0,...,\kappa\},
$$
where $d_{j,p}=d_{p,j}$.
\end{enumerate}
\end{thm}
\begin{proof}
The sufficiency follows from Theorem \ref{Cself}.

To prove the necessity, suppose that $T=\calC_{a,b,c}T^*\calC_{a,b,c}$.

First, we show that $T_{\max}$ is $\calC_{a,b,c}$-selfadjoint. Indeed, since $T\preceq T_{\max}$, we have
$$
T_{\max}^*\preceq T^*=\calC_{a,b,c}T\calC_{a,b,c}\preceq\calC_{a,b,c}T_{\max}\calC_{a,b,c},
$$
which implies, due to the involutivity of $\calC_{a,b,c}$, that
$$
\calC_{a,b,c}T_{\max}^*\calC_{a,b,c}\preceq T_{\max}.
$$
Proposition \ref{CT*CKz} shows that kernel functions always belong to $\text{dom}(\calC_{a,b,c}T_{\max}^*\calC_{a,b,c})$, and so,
$$
\calC_{a,b,c}T_{\max}^*\calC_{a,b,c}K_z =T_{\max}K_z,\quad\forall z\in\C.
$$
By Proposition \ref{=>s}, we reach conclusion (2), and hence, by Theorem \ref{Cself}, the operator $T_{\max}$ is $\calC$-selfadjoint with respect to the conjugation $\calC_{a,b,c}$.

Thus, assertion (1) follows from the following inclusions
$$
\calC_{a,b,c}T\calC_{a,b,c}\preceq \calC_{a,b,c}T_{\max}\calC_{a,b,c}=T_{\max}^*\preceq T^*=\calC_{a,b,c}T\calC_{a,b,c}.
$$
\end{proof}

A similar situation is encountered by imposing the self-adjointness constraint.
\begin{thm}\label{no-max-do-self}
Let $T$ be a $\kappa$-order differential operator induced by the entire function coefficients $(\psi_j)_{j=0}^\kappa$ (note that $T\preceq T_{\max}$). The operator $T$ is selfadjoint if and only if the following conditions hold.
\begin{enumerate}
\item $T=T_{\max}$.
\item The symbols are of forms
$$
\psi_p(z)=\sum_{j=0}^\kappa d_{j,p}z^j,\quad \forall p\in\{0,...,\kappa\},
$$
where $d_{j,p}=\overline{d_{p,j}}$.
\end{enumerate}
\end{thm}
\begin{proof}
The sufficiency follows from Theorem \ref{selfadjoint}.

To prove the necessity, suppose that $T=T^*$.

First, we show that $T_{\max}$ is also selfadjoint. Indeed, since $T\preceq T_{\max}$, we have
$$
T_{\max}^*\preceq T^*=T\preceq T_{\max},
$$
which implies that
$$
T_{\max}^*\preceq T_{\max}.
$$
Proposition \ref{propT*K_z} shows that kernel functions always belong to $\text{dom}(T_{\max}^*)$, and so,
$$
T_{\max}^*K_z =T_{\max}K_z,\quad\forall z\in\C.
$$
By Proposition \ref{=>s}, we get conclusion (2), and hence, by Proposition \ref{selfadjoint}, the operator $T_{\max}$ is selfadjoint.

Thus, conclusion (1) follows from the following inclusions
$$
T\preceq T_{\max}=T_{\max}^*\preceq T^*=T.
$$
\end{proof}
\begin{rem}
Theorems \ref{no-max-do}-\ref{no-max-do-self} show that there is no nontrivial domain for a differential operator $T$ on which $T$ is $\calC$-selfadjoint with respect to a weighted composition conjugation, and selfadjoint, respectively.
\end{rem}

\section{Dictionary between Fock space and Lebesgue space}
This section is motivated by Stone-von Neumann theorem, which states that there is a unitary operator from Lebesgue space $L^2(\R)$ onto Fock space $\calF^2(\C)$ that intertwines the harmonic oscillator with the Euler operator. This specific unitary equivalence is rather explicit, implemented by the \emph{Segal-Bargmann transform}. 

We first adopt some notational conventions. On the space $L^2(\R, dx)$ we set
$$
Xf(x)=xf(x),\quad Pf(x)=-if'(x),\quad\forall x\in\R,\forall f\in L^2(\R),
$$
while on the Fock space $\calF^2(\C)$, we use the expressions
$$
\X g(z)=zg(z),\quad\mathbb{P}g(z)=g'(z),\quad\forall z\in\C,\forall g\in\calF^2(\C).
$$
The explicit form of the Segal-Bargmann transform is contained in the next Proposition.
\begin{prop}[{\cite[Theorem 6.2]{hall260holomorphic}}]
Consider the operator $U$ given by
\begin{equation}
Uf(z)=\int\limits_{\R}A(z,x)f(x)dx,\quad\forall f\in L^2(\R),\forall z\in\C,
\end{equation}
where $A:\C\times\R\to\C$
$$
A(z,x)=\pi^{-1/4}e^{(-z^2+2\sqrt{2}xz-x^2)/2}.
$$
Then:
\begin{enumerate}
\item The operator $U$ is unitary from $L^2(\R)$ onto $\calF^2(\C)$.
\item The identities
\begin{equation}\label{sb-1}
U\left(\dfrac{X+iP}{\sqrt{2}}\right)U^{-1}=\mathbb{P}
\end{equation}
\begin{equation}\label{sb-2}
U\left(\dfrac{X-iP}{\sqrt{2}}\right)U^{-1}=\X
\end{equation}
hold.
\end{enumerate}
\end{prop}

In the case of polynomial full symbols, Segal-Bargmann's transform maps linear differential operators on $\calF^2(\C)$ to linear differential operators acting on $L^2(\R)$.
\begin{thm}\label{transfer-diff-op}
Let $T_{\max}$ be the maximal differential operator on $\calF^2(\C)$ induced by
$$
\psi_p(z)=\sum_{j=0}^{\kappa}\alpha_{j,p}z^j,\quad\forall p\in\{0,\cdots,\kappa\}.
$$
Furthermore, let $L_{\max}$ be the maximal differential operator on $L^2(\R)$ induced by the following expression
$$
\sum_{p=0}^{\kappa}\sum_{j=0}^{\kappa}\alpha_{j,p}\left(\dfrac{X-iP}{\sqrt{2}}\right)^j\left(\dfrac{X+iP}{\sqrt{2}}\right)^p.
$$
Then we have the following conclusions
$$
\calD[\psi]=U\left[\sum_{p=0}^{\kappa}\sum_{j=0}^{\kappa}\alpha_{j,p}\left(\dfrac{X-iP}{\sqrt{2}}\right)^j\left(\dfrac{X+iP}{\sqrt{2}}\right)^p\right]U^{-1};
$$
$$
T_{\max}=UL_{\max}U^{-1}.
$$
\end{thm}
\begin{proof}
It is clear that the second conclusion follows from the first one.

Prove the first conclusion as follows. For every $f\in\calF^2(\C)$, we have
\begin{eqnarray*}
\calD[\psi]f(z) &=& \sum_{p=0}^{\kappa}\sum_{j=0}^{\kappa}\alpha_{j,p}z^j f^{(p)}(z)\\
&=&\sum_{p=0}^{\kappa}\sum_{j=0}^{\kappa}\alpha_{j,p}\X^j\mathbb{P}^p f(z)\\
&=&\sum_{p=0}^{\kappa}\sum_{j=0}^{\kappa}\alpha_{j,p}U\left(\dfrac{X-iP}{\sqrt{2}}\right)^j\left(\dfrac{X+iP}{\sqrt{2}}\right)^pU^{-1},
\end{eqnarray*}
which gives the desired conclusion.
\end{proof}

In general, the Segal-Bargmann transform does not preserve the order of a differential operator. One substantial advantage being that in particular cases the order of a differential operator can be decreased via this transform.
\begin{exa}
We start with the quantum oscillator, acting in the wrong space.
Let $H_{\max}$ be the maximal differential operator on $\calF^2(\C)$ given by
$$ (H_{\max}f) (z) = - f''(z) + z^2 f(z).$$
This operator is unitarily equivalent to the maximal differential operator $L_{\max}$ on $L^2(\R)$ given by
$$
(L_{\max}g)(x)=-g(x)-2xg'(x).
$$
\end{exa}

Note that identities \eqref{sb-1}-\eqref{sb-2} can be rewritten as
$$
P=U^{-1}\left(\dfrac{i(\mathbb{X}-\mathbb{P})}{\sqrt{2}}\right)U,\quad X=U^{-1}\left(\dfrac{\mathbb{X}+\mathbb{P}}{\sqrt{2}}\right)U.
$$

\begin{exa}\label{HarmOsc}
And now the same basic operator, but in the correct Lebesgue space setting.
Let $L_{\max}$ be the maximal differential operator on $L^2(\R)$ given by
$$ (L_{\max}f) (x) = - f''(x) + x^2 f(x).$$
Then it is unitarily equivalent to the maximal differential operator $V_{\max}$ on $\calF^2(\C)$ given by
$$
(V_{\max}g)(z)=g(z)+2zg'(z).
$$
\end{exa}

Two natural conjugations on the space $L^2(\R)$ stand aside in quantum field theory. The first conjugation, known as the \emph{time-reversal operator}, is defined by
$$
\mathcal{T}f(x)=\overline{f(x)},\quad\forall f\in L^2(\R),
$$
and the other one is 
$$
\mathcal{PT}f(x)=\overline{f(-x)},\quad\forall f\in L^2(\R).
$$
It can be seen that the conjugation $\mathcal{PT}$ is the product of the time-reversal operator $\mathcal{T}$ and the \emph{parity operator} $\mathcal{P}$ (defined as $\mathcal{P}f(x)=f(-x)$).

The last decade has witnessed a growing interest in the class of $\mathcal{PT}$-symmetric operators, due to a series of remarkable observations in the emerging field  of non-Hermitian quantum theory \cite{MR1686605, Caliceti, GPP, Henry, Langer-Tretter, Zihar, Znojil}. Below we confine ourselves to identify the form of $\mathcal{PT}$-symmetric operators on Fock space.

It turns out that the conjugations $\mathcal{T}$, $\mathcal{PT}$ are unitarily equivalent to simple weighed composition conjugations studied in the present article.
\begin{prop}\label{small}
The following identity holds
$$
U\mathcal{T}U^{-1}=\calC_{1,0,1},\quad U\mathcal{PT}U^{-1}=\calC_{-1,0,1}.
$$
\end{prop}
\begin{proof}
First we note that $A(z,\cdot)\in L^2(\R)$, and furthermore,
\begin{eqnarray}\label{impor-lunar-year}
\int\limits_{\R}A(z,x)A(u,x)dx=e^{zu},\quad\forall z,u\in\C.
\end{eqnarray}
The identity above gives
$$
U\left(A(\overline{w},\cdot)\right)=K_w,\quad\forall w\in\C.
$$
Using the above identity, we have
$$
\mathcal{T}U^{-1}K_w=\mathcal{T}\left(A(\overline{w},\cdot)\right)=A(w,+\cdot),
$$
and
$$
\mathcal{PT}U^{-1}K_w=\mathcal{PT}\left(A(\overline{w},\cdot)\right)=A(w,-\cdot).
$$
Therefore,
$$
U\mathcal{T}U^{-1}K_w=U\left(A(w,+\cdot)\right)=U\left(\overline{A(\overline{w},\cdot)}\right)=K_{\overline{w}}=\calC_{1,0,1}K_w;
$$
$$
U\mathcal{PT}U^{-1}K_w=U\left(A(w,-\cdot)\right)=U\left(A\left(\overline{-\overline{w}},\cdot\right)\right)=K_{-\overline{w}}=\calC_{-1,0,1}K_w.
$$
Since the operators $U,\mathcal{P},\mathcal{T}$ are bounded and the set of kernel functions is dense in $\calF^2$, we reached the conclusion.
\end{proof}

\begin{rem}
The above proposition and Theorem \ref{Cself} show that a linear differential operator $R$ is $\mathcal{PT}$-selfadjoint on Lebesgue space if and only if the operator $URU^{-1}$ is a linear combination 
with complex coefficients of
operators of the form
$$ z^j [\frac{\partial}{\partial z}]^p + (-1)^{j+p}  z^p [\frac{\partial}{\partial z}]^j.$$
The sum is of course endowed with the maximal domain of definition.
\end{rem}

The nature of the general conjugations $U^{-1}\calC_{a,b,c}U$ acting on Lebesgue space $L^2(\R)$ is revealed, on basis vectors, in the following Proposition.

\begin{prop}
Let $\calC_{a,b,c}$ be a weighted composition conjugation on $\calF^2(\C)$. Then the identity
$$
U^{-1}\calC_{a,b,c}U(h_{0,m})=\sum_{k=0}^m\binom{m}{k}ca^k b^{m-k}h_{b,k},
$$
holds, where 
$$
h_{w,m}(x)=\dfrac{\partial^m}{\partial u^m}A(u,x)\bigg\arrowvert_{u=w}.
$$
\end{prop}
\begin{proof}
Differentiating \eqref{impor-lunar-year} $m$ times and evaluating it at the point $u=w$, we get
\begin{eqnarray}\label{impor-lunar-year-1}
\int\limits_{\R}A(z,x)\dfrac{\partial^m}{\partial u^m}A(u,x)\bigg\arrowvert_{u=w} dx=z^me^{zw},\quad\forall z,w\in\C,\forall m\in\N,
\end{eqnarray}
and hence,
$$
U(h_{w,m})(z)=z^me^{zw},\quad\forall z,w\in\C,\forall m\in\N.
$$
Taking into account the structure of the conjugation $\calC_{a,b,c}$, we obtain
\begin{eqnarray*}
\calC_{a,b,c}U(h_m)(z)&=& ce^{bz}(az+b)^m=\sum_{k=0}^m\binom{m}{k}ca^k b^{m-k}z^ke^{bz}\\
&=&\sum_{k=0}^m\binom{m}{k}ca^k b^{m-k}U(h_{b,k})(z),
\end{eqnarray*}
which gives the desired result.
\end{proof}

\section{Eigenvalues evaluation}
In general, there is no common method to compute the spectrum of all higher-order differential operators. First note that in the case when the total symbols are polynomials, our Theorem \ref{transfer-diff-op} can offer an effective method. It allows us to link the current paper to $\mathcal{PT}$-symmetric operators proposed in quantum mechanics.

Next, we focus only on the very restrictive category of linear differential operators of the form:
\begin{equation}\label{eq-higher}
Vf=\psi_0 f+\psi_n f^{(n)},\quad f\in\text{dom}(V),
\end{equation}
where $\psi_0,\psi_n$ are entire functions.

\begin{prop}\label{V-point-spec}
Let $V$ be an $n$th-order differential operator, induced by the symbols above. If the function $\psi_n$ has a zero at $w$ of order $n$, then
$$
\sigma_p(V)\subseteq\left\{\psi_0(w),\,\,\psi_0(w)+\binom{k}{n}\psi_n^{(n)}(w):k\in\N_{\geq n}\right\}.
$$
\end{prop}
\begin{proof}
Take arbitrarily $\lambda\in\sigma_p(V)$. Then there exists $f\in\text{dom}(V)\setminus\{\mathbf{0}\}$ such that
\begin{equation}\label{eq-wed-1}
\psi_0(z)f(z)+\psi_n(z)f^{(n)}(z)=Vf(z)=\lambda f(z),\quad\forall z\in\C.
\end{equation}

If $f(w)\ne 0$, then
$$
\psi_0(w)f(w)=\psi_0(w)f(w)+\psi_n(w)f^{(n)}(w)=\lambda f(w),
$$
which gives $\lambda=\psi_0(w)$.

Suppose that $f$ has a zero at $w$ of order $\ell\geq 1$. Then differentiating \eqref{eq-wed-1} $\ell$ times and evaluating it at the point $z=w$ yields
\begin{equation}\label{eq-monday-1}
\sum_{j=0}^\ell\binom{\ell}{j} \psi_0^{(\ell-j)}(w)f^{(j)}(w)+\sum_{j=n}^{\ell+n}\binom{\ell}{j-n} \psi_n^{(\ell+n-j)}(w)f^{(j)}(w)=\lambda f^{(\ell)}(w).
\end{equation}
If $\ell < n$, then
$$
\psi_n^{(\ell+n-j)}(w)=0,\quad\forall j\in\{n,...,\ell+n\},
$$
and hence, equation \eqref{eq-monday-1} is reduced to
$$
\psi_0(w)f^{(\ell)}(w)=\lambda f^{(\ell)}(w),
$$
which gives $\lambda=\psi_0(w)$ (since $f^{(\ell)}(w)\ne 0$).

If $\ell \geq n$, then equation \eqref{eq-monday-1} is reduced to
$$
\psi_0(w)f^{(\ell)}(w)+\sum_{j=\ell}^{\ell+n}\binom{\ell}{j-n} \psi_n^{(\ell+n-j)}(w)f^{(j)}(w)=\lambda f^{(\ell)}(w),
$$
implying
$$
\psi_n^{(\ell+n-j)}(w)=0,\quad\forall j\in\{\ell+1,...,\ell+n\},
$$
and
$$
\psi_0(w)f^{(\ell)}(w)+\binom{\ell}{\ell-n}\psi_n^{(n)}(w)f^{(\ell)}(w)=\lambda f^{(\ell)}(w).
$$
In short we obtain the explicit form of the eigenvalue:
$$
\lambda=\psi_0(w)+\binom{\ell}{\ell-n}\psi_n^{(n)}(w).
$$
\end{proof}

A consequence of Proposition \ref{propT*K_z} follows.
\begin{lem}
Always $K_z^{[m]}\in\text{dom}(V)$, and moreover,
\begin{enumerate}
\item if $m\geq n$, then
\begin{eqnarray*}
V^*K_z^{[m]}&=&\sum_{j=0}^{n-1}\binom{m}{j}\overline{\psi_0^{(m-j)}(z)}K_z^{[j]}\\
&&+\sum_{j=n}^m\left[ \binom{m}{j}\overline{\psi_0^{(m-j)}(z)}+\binom{m}{j-n}\overline{\psi_n^{(m-j+n)}(z)} \right]K_z^{[j]}\\
&&+\sum_{j=m+1}^{m+n}\binom{m}{j-n}\psi_n^{(m-j+n)}(z)K_z^{[j]};
\end{eqnarray*}
\item if $m< n$, then
\begin{eqnarray*}
V^*K_z^{[m]}&=&\sum_{j=0}^m\binom{m}{j}\overline{\psi_0^{(m-j)}(z)}K_z^{[j]}+\sum_{j=n}^{m+n}\binom{m}{j-n}\overline{\psi_n^{(m-j+n)}(z)}K_z^{[j]}.
\end{eqnarray*}
\end{enumerate}
\end{lem}

With the above preliminary remarks we can turn to the point spectrum of the adjoint.
\begin{prop}\label{V*-point-spec}
Let $V$ be an $n$th-order differential operator (being closed), defined by \eqref{eq-higher}. Suppose that $V$ is densely defined. If the function $\psi_n$ has a zero at $w$ of order $n$, then
$$\overline{\psi_0(w)},\,\,\overline{\psi_0(w)}+\binom{k}{n}\overline{\psi_n^{(n)}(w)},$$
where $k\in\N_{\geq n}$, are eigenvalues of $V^*$.
\end{prop}
\begin{proof}
Let $m\geq n$ denote by $\mathbf{K}_m$ the linear span of $\{K_w, K_w^{[1]},...,K_w^{[m]}\}$. Since the function $\psi_n$ has a zero at $w$ of order $n$,
\begin{equation*}
V^*K_z^{[p]}=
\begin{cases}
\sum_{j=0}^{n-1}\binom{p}{j}\overline{\psi_0^{(p-j)}(z)}K_z^{[j]}+\sum_{j=n}^p\left[ \binom{p}{j}\overline{\psi_0^{(p-j)}(z)}+\binom{p}{j-n}\overline{\psi_n^{(p-j+n)}(z)} \right]K_z^{[j]},\quad\text{if $p\geq n$,}\\
\sum_{j=0}^p\binom{p}{j}\overline{\psi_0^{(p-j)}(z)}K_z^{[j]},\quad\text{if $p<n$.}
\end{cases}
\end{equation*}
The above identities show that the matrix representing the operator $V^*$ restricted to $\mathbf{K}_m$, with respect to the above basis, is
\begin{displaymath}
A_m =
\left( \begin{array}{cccc}
D_{n} & * & \ldots & *\\
0 & \overline{\psi_0(w)+\psi_n^{(n)}(w)} & \ldots & *\\
0 & 0 & \ldots & * \\
\vdots & \vdots & \ddots & \vdots\\
0 & 0 & \ldots & \overline{\psi_0(w)+\binom{m}{m-n}\psi_n^{(n)}(w)}
\end{array} \right),
\end{displaymath}
where $D_n$ is the diagonal matrix with size $n$, in which the main diagonal entries are $\overline{\psi_0(w)}$. The subspace $\mathbf{K}_m$ is finite dimensional and therefore is closed. Fock space can be decomposed as
$$
\calF^2(\C)=\mathbf{K}_m\oplus\mathbf{K}_m^\bot.
$$
The block matrix of $V^*$ with respect to this decomposition is
\begin{displaymath}
\left( \begin{array}{cc}
A_m & B_m\\
0 & C_m
\end{array} \right).
\end{displaymath}
Thus, the spectrum of $V^*$ is the union of the spectrum of $A_m$ and the spectrum of $C_m$, and hence, we conclude that the spectrum of $V^*$ contains the following set
$$
\left\{\overline{\psi_0(w)},\,\,\overline{\psi_0(w)}+\binom{k}{k-n}\psi_n^{(n)}(w):k\in\{n,n+1,...,m\}\right\}
$$
But $m$ is arbitrary, therefore we infer that any complex number of the form
$$\overline{\psi_0(w)}+\binom{k}{k-n}\overline{\psi_n^{(n)}(w)},$$
where $k$ is a non-negative integer $\geq n$, is the eigenvalue of $V^*$.
\end{proof}

A combination of the two preceding results above gives the main result of this section. The proof below is based on the fact that $\lambda$ is an eigenvalue of a $\calC$-sefladjoint operator if and only if $\overline{\lambda}$ is an eigenvalue of the adjoint operator.
\begin{thm}\label{general}
Let $V$ be a $\calC$-selfadjoint differential operator (note that $V\preceq V_{\max}$ and $\calC$ is an arbitrary conjugation), defined by
$$
Vf=\psi_0 f+\psi_n f^{(n)},\quad f\in\text{dom}(V),
$$
where $\psi_0,\psi_n$ are entire functions. If the function $\psi_n$ has a zero at $w$ of order $n$, then
$$\sigma_p(V)=\left\{\overline{\psi_0(w)},\,\,\overline{\psi_0(w)}+\binom{k}{n}\overline{\psi_n^{(n)}(w)},\quad k\in\N_{\geq n}\right\}.$$
\end{thm}

The particular situation of complex self-adjoint operators with symmetry belonging to the three parameter family of weighted composition operators follows.

\begin{cor}
Let $V$ be a $\calC_{a,b,c}$-selfadjoint differential operator, defined by
$$
Vf=\psi_0 f+\psi_n f^{(n)},\quad f\in\text{dom}(V),
$$
where 
$$
\psi_0(z)=\alpha,\quad\psi_n(z)=(az+b)^n.
$$
Then
$$\sigma_p(V)=\left\{\overline{\alpha},\,\,\overline{\alpha}+\binom{k}{n}n!,\quad k\in\N_{\geq n}\right\}.$$
\end{cor}

Although a selfadjoint operator is always necessarily $\calC$-selfadjoint, their properties are completely different, e.g., the spectrum of a selfadjoint operator is real, but it can be complex in the case of $\calC_{a,b,c}$-selfadjoint.

\subsection{First-order linear differential operators}

Even the simplest case of first order differential operators offers a rich ground for further exploration. 
Let $T$ be a first-order differential operator, induced by $(\psi_j)_{j=0}^1$. The maximal differential operator of $T$ is denoted by $T_{\max}$. If $T$ is densely defined then for every $z\in\C$, $m\in\N$, we have $K_z^{[m]}\in\text{dom}(T^*)$, and moreover,
\begin{eqnarray}\label{T*K_z}
T^*K_z^{[m]}&=&\overline{\psi_0^{(m)}(z)}K_z+\sum_{k=0}^{m-1}\left[\binom{m}{k}\overline{\psi_0^{(k)}(z)}+\binom{m}{k+1}\overline{\psi_1^{(k+1)}(z)}\right]K_z^{[m-k]}\nonumber\\
&&+\overline{\psi_1(z)}K_z^{[m+1]}.
\end{eqnarray}

In order to analyze the spectrum of the operators $T,T_{\max}$, we split the discussion according to a dichotomy referring to the vanishing of the top coefficient $\psi_1$.

\subsubsection{The case when $\psi_1$ has zeroes}
In view of Proposition \ref{V-point-spec} we state the following observation.
\begin{prop}\label{S-point-spec}
Let $T$ be a first-order differential operator, induced by $(\psi_j)_{j=0}^1$. If $\psi_1$ has a zero at $w$ of order $1$, then
$$
\sigma_p(T)\subseteq\{\psi_0(w)+k\psi_1^{'}(w):k\in\N\}.
$$
\end{prop}
 One step further,
 
\begin{prop}\label{S*-point-spec}
Let $T$ be a first-order differential operator, induced by $(\psi_j)_{j=0}^1$. Suppose that $T$ is densely defined. If the function $\psi_1$ has a zero at $w$ of order $1$, then $\overline{\psi_0(w)+k\psi_1^{'}(w)}$, where $k\in\N$, are eigenvalues of $T^*$.

In particular, the eigenfunctions, corresponding to eigenvalues $\overline{\psi_0(w)}, \overline{\psi_0(w)+\psi_1^{'}(w)}$, are $K_w$, respectively 
\begin{equation*}
\eta(z)=
\begin{cases}
K_w^{[1]}(z),\quad\text{if $\psi_0^{'}(w)=0$,}\\
K_w(z)+\dfrac{\overline{\psi_1^{'}(w)}}{\overline{\psi_0^{'}(w)}}K_w^{[1]}(z),\quad\text{otherwise}.
\end{cases}
\end{equation*}

\end{prop}
\begin{proof}
For the first conclusion we invoke Proposition \ref{V*-point-spec}.

It remains to verify the second conclusion. Note that, also by \eqref{T*K_z} and $\psi_1(w)=0$,
\begin{eqnarray*}
T^*K_w^{[1]}&=&\overline{\psi_0^{'}(w)}K_w+\overline{\psi_0(w)+\psi_1^{'}(w)}K_w^{[1]}+\overline{\psi_1(w)}K_w^{[2]}\\
&=&\overline{\psi_0^{'}(w)}K_w+\overline{\psi_0(w)+\psi_1^{'}(w)}K_w^{[1]}.
\end{eqnarray*}
If $\psi_0^{'}(w)=0$, then the identity above shows that
$$
T^*K_w^{[1]}=\overline{\psi_0(w)+\psi_1^{'}(w)}K_w^{[1]}.
$$
If $\psi_0^{'}(w)\ne 0$, then
\begin{eqnarray*}
T^*\eta &=& T^*K_w+\dfrac{\overline{\psi_1^{'}(w)}}{\overline{\psi_0^{'}(w)}}T^* K_w^{[1]}\\
&=&\overline{\psi_0(w)}K_w+\overline{\psi_1^{'}(w)}K_w+\dfrac{\overline{\psi_1^{'}(w)}(\overline{\psi_1^{'}(w)}+\overline{\psi_0(w)})}{\overline{\psi_0^{'}(w)}}K_w^{[1]}\\
&=&[\overline{\psi_0(w)+\psi_1^{'}(w)}]\eta.
\end{eqnarray*}
\end{proof}

\subsubsection{The coefficient $\psi_1$ has no zeros}
\begin{lem}\label{kerS}
Let $\psi_0,\psi_1$ be entire functions such that $\psi_1$ is never vanished. Suppose that the function $f$ satisfies
$$
\psi_0(z)f(z)+\psi_1(z)f^{'}(z)=g(z).
$$
Then the function $f$ must be of the following form
$$
f(z)=\left[f(0)+\int_0^z\dfrac{g(y)}{\psi_1(y)}\exp\left(\int_0^y \dfrac{\psi_0(x)}{\psi_1(x)}dx\right)dy\right]\exp\left(-\int_0^z \dfrac{\psi_0(x)}{\psi_1(x)}dx\right).
$$
In particular,
\begin{enumerate} 
\item If $g\equiv\mathbf{0}$, then
$$
f(z)=f(0)\exp\left(-\int_0^z \dfrac{\psi_0(x)}{\psi_1(x)}dx\right).
$$
\item If $g=\psi_0$, then
$$
f(z)=(f(0)-1)\exp\left(-\int_0^z\dfrac{\psi_0(x)}{\psi_1(x)}dx\right)+1.
$$
\end{enumerate}
In both cases, $f\in\calF^2(\C)$ if and only if the functions $\psi_0,\psi_1$ satisfy
\begin{equation*}\label{g=0=ker}
\psi_0(z)=(\alpha z+\beta)\psi_1(z),\quad\text{with $|\alpha|<1$.}
\end{equation*}
\end{lem}
\begin{proof}
Putting
$$
h(z)=f(z)\exp\left(\int_0^z \dfrac{\psi_0(x)}{\psi_1(x)}dx\right),
$$
then
\begin{eqnarray*}
f'(z)&=&\left(h(z)\exp\left(-\int_0^z \dfrac{\psi_0(x)}{\psi_1(x)}dx\right)\right)^{'}\\
&=& h'(z)\exp\left(-\int_0^z \dfrac{\psi_0(x)}{\psi_1(x)}dx\right)+h(z)\left(-\dfrac{\psi_0(z)}{\psi_1(z)}\right)\exp\left(-\int_0^z \dfrac{\psi_0(x)}{\psi_1(x)}dx\right)\\
&=& h'(z)\exp\left(-\int_0^z \dfrac{\psi_0(x)}{\psi_1(x)}dx\right)+\left(-\dfrac{\psi_0(z)}{\psi_1(z)}\right)f(z),
\end{eqnarray*}
which gives
$$
h'(z)\exp\left(-\int_0^z \dfrac{\psi_0(x)}{\psi_1(x)}dx\right)=\dfrac{\psi_0(z)f(z)+\psi_1(z)f^{'}(z)}{\psi_1(z)}=\dfrac{g(z)}{\psi_1(z)}.
$$
Thus,
$$
h(z)=h(0)+\int_0^z\dfrac{g(y)}{\psi_1(y)}\exp\left(\int_0^y \dfrac{\psi_0(x)}{\psi_1(x)}dx\right)dy.
$$

(1) If $g\equiv\mathbf{0}$, then $h(z)=h(0)=f(0)$ and hence, we get the desired result.

(2) If $g=\psi_0$, then by arguments similar to those used in the first part but now applied to the equation
$$
\psi_0(z)(f(z)-1)+\psi_1(z)f'(z)=0,
$$
we get the desired result.

The rest conclusion follows from \cite[Theorem 1.1]{KHI}.
\end{proof}

\begin{prop}\label{sigma_p-no}
Let $T$ be a first-order differential operator, induced by entire functions $\psi_0,\psi_1$, where $\psi_1$ is zero free. The following assertions hold:
\begin{enumerate}
\item If $\lambda\in\sigma_p(T)$, then
\begin{equation}\label{g=0=ker}
\psi_0(z)=(\alpha z+\beta)\psi_1(z)+\lambda,\quad\text{with $\abs{\alpha}<1$.}
\end{equation}
\item The point spectrum $\sigma_p(T_{\max})\ne\emptyset$ if and only if condition \eqref{g=0=ker} is satisfied. In this case, $\lambda\in\sigma_p(T_{\max})$.
\end{enumerate}
\end{prop}
\begin{proof}
$(1)$ Since $\lambda\in\sigma_p(T)$, there exists $f\in\text{dom}(T)\setminus\{\mathbf{0}\}\subseteq\calF^2(\C)$ such that $Tf=\lambda f$, which gives
$$
(\psi_0(z)-\lambda)f(z)+\psi_1(z)f'(z)=0,\quad\forall z\in\C.
$$
Lemma \ref{kerS} shows that
$$
f(z)=f(0)\exp\left(-\int_0^z\dfrac{\psi_0(x)-\lambda}{\psi_1(x)}\,dx\right),
$$
where $f(0)\ne 0$ (since $f\not\equiv\mathbf{0}$). By \cite[Theorem 1.1]{KHI}, this function belongs to $\calF^2(\C)$ if and only if we have \eqref{g=0=ker}.

$(2)$ The necessity follows from the first assertion. To prove the sufficiency, we suppose that \eqref{g=0=ker} holds, where $\alpha,\beta,\lambda$ are constants with $\abs{\alpha}<1$. With this condition, the function $g$, given by
$$
g(z)=e^{-\alpha z^2/2-\beta z}\in\calF^2(\C)\quad\text{(by \cite[Theorem 1.1]{KHI})},
$$
satisfies the following identity
$$
\psi_0 g+\psi_1g'=\lambda g \in\calF^2(\C).
$$
This shows that $g\in\text{dom}(T_{\max})$ and $T_{\max}g=\lambda g$, i.e. $\lambda\in\sigma_p(T_{\max})$.
\end{proof}

\begin{prop}\label{p-sigma>=2}
Let $T$ be a first-order differential operator, induced by entire functions $\psi_0,\psi_1$, where $\psi_1$ is never vanished. If $\emph{Card}(\sigma_p(T))\geq 2$, then the symbols $\psi_0,\psi_1$ are of the following forms
\begin{equation}\label{>=2}
\psi_0(z)=Az+B, \quad\psi_1(z)=C\ne 0,\,\, \forall z\in\C,
\end{equation}
with $|A| < |C|$.
\end{prop}
\begin{proof}
Suppose that $\lambda_1,\lambda_2\in\sigma_p(T)$, where $\lambda_1\ne\lambda_2$. By Proposition \ref{sigma_p-no}, there exist $\alpha_j,\beta_j$, $j\in\{1,2\}$, such that
$$
\psi_0(z)=(\alpha_j z+\beta_j)\psi_1(z)+\lambda_j,\quad\text{with $|\alpha_j|<1$}.
$$
Hence,
$$
[(\alpha_1-\alpha_2)z+\beta_1-\beta_2]\psi_1(z)=\lambda_2-\lambda_1,
$$
which implies, since $\psi_1$ is an entire function, that $\alpha_1=\alpha_2$, and consequently, $\psi_1$ is a constant function. By \eqref{g=0=ker}, we get \eqref{>=2}.
\end{proof}

For maximal differential operators, the converse implication in Proposition \ref{p-sigma>=2} is true. We state the following result.
\begin{prop}\label{p-sigma>=2-next}
Let $T_{\max}$ be a maximal differential operator of order $1$, induced by the symbols \eqref{>=2}, that is
$$
\psi_0(z)=Az+B, \quad\psi_1(z)=C\ne 0,\,\, \forall z\in\C.
$$ 
Then
\begin{equation*}
\sigma_p(T_{\max})=
\begin{cases}
\C,\quad\text{if $|A|<|C|$,}\\
\emptyset,\quad\text{otherwise.}
\end{cases}
\end{equation*}
\end{prop}
\begin{proof}
By Proposition \ref{sigma_p-no}, $\sigma_p(T_{\max})=\emptyset$ if $|A|\geq |C|$.

For the remaining conclusion, we take an arbitrary point $\zeta\in\C$. Since $|A|<|C|$, by \cite[Theorem 1.1]{KHI}, the function $g$, given by
$$
g(z)=\exp\left(-\dfrac{Az^2}{2C}-\dfrac{B-\zeta}{C}z\right),
$$
always belongs to $\calF^2(\C)$. A direct computation shows that
$$
\psi_0 g+\psi_1 g'=\zeta g \in\calF^2(\C),
$$
which gives $g\in\text{dom}(T_{\max})$ and $T_{\max}g=\zeta g$. By the definition of the point spectrum, $\zeta\in\sigma_p(T_{\max})$.
\end{proof}


As an application we adapt the above general remarks to the case of $\calC$-selfadjoint operators.

\begin{cor}\label{final}
Let $T$ be a $\calC_{a,b,c}$-selfadjoint differential operator of order $1$, defined by the coefficients $(\psi_j)_{j=0}^1$, that is $T=T_{\max}$ and
$$
\psi_0(z)=d_{0,0}+d_{1,0}(az+b),\quad\psi_1(z)=d_{0,1}+d_{1,1}(az+b),\quad\text{where $d_{1,0}=d_{0,1}$}.
$$
Then
\begin{equation*}
\sigma_p(T)=
\begin{cases}
\left\{\psi_0\left(\frac{-d_{0,1}-bd_{1,1}}{ad_{1,1}}\right)+k\psi_1^{'}\left(\frac{-d_{0,1}-bd_{1,1}}{ad_{1,1}}\right):k\in\N\right\},\quad\text{if $d_{1,1}\ne 0$,}\\
\emptyset,\quad\text{if $d_{1,1}=0$, $d_{0,1}\ne 0$,}\\
\{d_{0,0}\},\quad\text{if $d_{1,1}=d_{0,1}=0$.}
\end{cases}
\end{equation*}
\end{cor}
\begin{proof}
The first identity holds in virtue of Theorem \ref{general}, while the second one follows from Proposition \ref{p-sigma>=2-next}. The proof for the last identity is elementary.
\end{proof}

\begin{rem}
We mention that the above framework applies to the simple quantum oscillator on Fock space $\calF^2(\C)$.  In view of Corollary \ref{final} we infer that the operator
$$
\widehat{H}f=f(z) + 2zf'(z),\quad\text{dom}(\widehat{H})=\left\{f\in\calF^2(\C): f + 2 zf' \in \calF^2(\C) \right\}
$$
is $\calC_{-1,0,1}$-selfadjoint on $\calF^2(\C)$. That is $\widehat{H}$ is $\mathcal{PT}$-selfadjoint. Furthermore, its point spectrum (and as a matter of fact, the entire spectrum) coincides with the odd positive integers
$$
\sigma_p(\widehat{H})=\left\{2k+1:k \geq 0 \right\}
$$
with eigenfunctions given by the monomials $z^k, \ \ k \geq 0$.
\end{rem}

\bibliographystyle{plain}
\bibliography{refs}
\end{document}